%% file: Theorem.tex
\documentclass[a4paper]{amsart}
\usepackage[utf8]{inputenc}
\usepackage{amsmath}
\usepackage{amsthm}
\usepackage{amsfonts,amssymb}
\usepackage[foot]{amsaddr}
\usepackage{mathtools}
\usepackage{mathshortcuts}
\usepackage{enumitem}
\usepackage[hypertexnames=false]{hyperref}
\usepackage{cleveref}
\usepackage{autonum}
\usepackage{float}
\usepackage{subcaption}
\usepackage{mleftright}
\usepackage[left=3cm,right=3cm,top=3.5cm,bottom=3.5cm]{geometry}
\usepackage{tikz}
\usetikzlibrary{decorations,snakes,decorations.markings,positioning,shadows.blur,calc,knots,shapes.misc,decorations.pathreplacing,calligraphy}
\usepackage{pgfplots}
\pgfplotsset{compat=newest}
\usepgfplotslibrary{fillbetween}
\usepackage{mdframed}

\allowdisplaybreaks

\newcommand{\ket}[1]{\lvert #1 \rangle}
\newcommand{\bra}[1]{\langle #1 \rvert}
\newcommand{\ketbra}[2]{\lvert #1 \rangle\!\langle #2 \rvert}

\newcommand{\mixed}[1]{\,\overline{\mathrm{id}}_{#1}}

\usepackage{tcolorbox}
\tcbuselibrary{theorems}

\theoremstyle{plain}
\newtheorem{thm}{Theorem}[section]
\newtheorem{lem}[thm]{Lemma}

\newtheorem{cor}[thm]{Corollary}

\newtheorem{innercustomthm}{Claim}
\newenvironment{customthm}[1]
{\renewcommand\theinnercustomthm{#1}\innercustomthm}
{\endinnercustomthm}

\theoremstyle{definition}
\newtheorem{defn}[thm]{Definition}

\newtheorem{rem}[thm]{Remark}
\newtheorem{notation}[thm]{Notation}

\definecolor{LightGray}{RGB}{220,220,220}
\definecolor{myred}{RGB}{255, 19, 0}
\definecolor{myblue}{RGB}{14, 81, 167}
\definecolor{myorange}{RGB}{255, 129, 0}
\definecolor{mygreen}{RGB}{0, 146, 44}

\hypersetup{
	colorlinks,
	linkcolor=myblue,
	citecolor=myblue,
	filecolor=myblue,
	urlcolor=myblue,
	pdftitle={Commuting operations factorise},
	pdfauthor={Renato Renner and Ramona Wolf}
}

\crefname{defn}{Definition}{Definitions}
\Crefname{defn}{Definition}{Definitions}
\crefname{notation}{Notation}{Notations}
\Crefname{notation}{Notation}{Notations}
\crefname{thm}{Theorem}{Theorems}
\Crefname{thm}{Theorem}{Theorems}
\crefname{lem}{Lemma}{Lemmas}
\Crefname{lem}{Lemma}{Lemmas}
\crefname{rem}{Remark}{Remarks}
\Crefname{rem}{Remark}{Remarks}
\crefname{prop}{Proposition}{Propositions}
\Crefname{prop}{Proposition}{Propositions}
\crefname{cor}{Corollary}{Corollaries}
\Crefname{cor}{Corollary}{Corollaries}
\crefname{section}{Section}{Sections}
\Crefname{section}{Section}{Sections}
\crefname{equation}{}{}
\Crefname{equation}{}{}
\crefname{figure}{Figure}{Figures}
\Crefname{figure}{Figure}{Figures}
\crefname{appendix}{Appendix}{Appendices}
\Crefname{appendix}{Appendix}{Appendices}
\crefname{claim}{Claim}{Claims}
\Crefname{claim}{Claim}{Claims}
\crefname{innercustomthm}{Claim}{Claims}
\Crefname{innercustomthm}{Claim}{Claims}

\labelcrefformat{subequation}{#2(#1)#3}
\labelcrefrangeformat{subequation}{#3(#1)#4 to #5(#2)#6}

\mleftright

\title{Commuting operations factorise}
\author{Renato Renner}
\email{renner@ethz.ch}
\author{Ramona Wolf}
\email{rawolf@phys.ethz.ch}

\address{Institute for Theoretical Physics, ETH Zurich, Zurich, Switzerland}
\address{Quantum Center, ETH Zurich, Zurich, Switzerland}

\begin{document}

\begin{abstract}
	Consider two agents, Alice and Bob, each of whom takes a quantum input, operates on a shared quantum system~$K$, and produces a quantum output. Alice and Bob's operations may commute, in the sense that the joint input-output behaviour is independent of the order in which they access~$K$. Here we ask whether this commutation property implies that $K$ can be split into two factors on which Alice and Bob act separately. The question can be regarded as a ``fully quantum'' generalisation of a problem posed by Tsirelson, who considered the case where Alice and Bob's inputs and outputs are classical. In this case, the answer is negative in general, but it is known that a factorisation exists in finite dimensions. Here we show the same holds in the fully quantum case, i.e., commuting operations factorise, provided that all input systems are finite-dimensional.
\end{abstract}

\maketitle


\input{1_introduction}

\input{1_preliminaries}
\input{2_theorem}

\input{3_independence}

\section*{Acknowledgements}
This work was supported by the Air Force Office of Scientific Research (AFOSR), grant No.\ FA9550-19-1-0202, the QuantERA project eDICT, the SNSF grant No.\  200021\_188541, the National Centre of Competence in Research SwissMAP, and the Quantum Center at ETH Zurich. We acknowledge the hospitality of the Centro de Ciencias de Benasque Pedro Pascual, Spain.


\appendix
\input{A_Choi}

\input{B_unital}

\bibliographystyle{halpha}
\bibliography{Theorembib}

\end{document}

%% file: 1_introduction.tex
\section{Introduction}


Let $K$ be a quantum system accessible to two agents, Alice and Bob, each operating on it once. Alice's operation $\mathcal{X}_{I, A}$ depends on an input~$I$ and generates an output~$A$. Similarly, Bob's operation $\mathcal{Y}_{J, B}$ depends on~$J$ and generates~$B$. We say that Alice and Bob's operations \emph{commute} if the order in which $\mathcal{X}_{I, A}$ and $\mathcal{Y}_{J, B}$ are applied has no influence on the outputs $A$ and $B$, and on how they depend on the inputs $I$ and $J$. Physical considerations often imply commutation.  For example, it holds if Alice and Bob's operations are executed at spacelike separation in a relativistic spacetime. Here we ask the following question:

\begin{quote}
\emph{Does commutation between the maps $\mathcal{X}_{I, A}$ and $\mathcal{Y}_{J, B}$ imply the existence of a factorisation of $K$ such that the maps act nontrivially solely on separate factors?}
\end{quote}

This paper aims to provide a precise formulation of this question for the generic case where the inputs and outputs $I$, $J$, $A$, and~$B$ are arbitrary quantum systems (cf.~\cref{fig:simplifiedscenario}). Our question can thus be understood as a ``fully quantum'' version of \emph{Tsirelson's problem}, which corresponds to the special case where these systems are all classical random variables (and these take values from finite sets). In the latter, $\mathcal{X}_{I, A}$ represents a measurement on~$K$ that depends on a classical choice, encoded in $I$, and outputs a classical result, encoded in~$A$, and likewise for $\mathcal{Y}_{J, B}$. Tsirelson's problem was initially posed in \cite{Tsirelson1993}. While it was prematurely claimed that it always has a positive answer, the proof, also by Tsirelson, assumes that $K$ is finite-dimensional~\cite{Tsirelson2006}. While this assumption can be relaxed~\cite{Scholz2008}, the answer to Tsirelson's problem is negative if it is dropped completely \cite{Ji2021} (see~\cite{Cabello2023} for an overview).

The main technical contribution of this work is a proof that answers the general question above affirmatively under the assumptions that the map $\mathcal{X}_{I, A}$ or the map $\mathcal{Y}_{J, B}$ is unital on $K$, which is automatically satisfied in Tsirelson's problem, and that the systems $I$, $J$, and $K$ are finite-dimensional.

\begin{figure}[t]
	\centering
	\begin{tikzpicture}[baseline=(current bounding box.center),scale=1.15,xscale=1.1]
		\draw[fill=gray!30] (0,0) rectangle (1,0.75);
		\draw[fill=gray!30] (0.7,1.25) rectangle (1.7,2);
		\node at (0.5,0.375) {\small $\mathcal{X}_{I, A}$};
		\node at (1.2,1.625) {\small $\mathcal{Y}_{J, B}$};
		\draw[thick,->,>=stealth] (0.15,-0.5) -- (0.15,0);
		\draw[thick,->,>=stealth] (0.85,-0.5) -- (0.85,0);
		\draw[thick,->,>=stealth] (0.15,0.75) -- (0.15,2.5);
		\draw[thick,->,>=stealth] (1.55,-0.5) -- (1.55,1.25);
		\draw[thick,->,>=stealth] (0.85,0.75) to node[right] {\small $K$} (0.85,1.25);
		\draw[thick,->,>=stealth] (0.85,2) to node[right] {\small $K$} (0.85,2.35);
		\draw[thick,->,>=stealth] (1.55,2) -- (1.55,2.5);
		\draw[very thick] (0.7,2.35) -- (1,2.35);
		\node at (0.15,2.65) {\small $A$};
		\node at (1.55,2.65) {\small $B$};
		\node at (0.15,-0.65) {\small $I$};
		\node at (0.85,-0.65) {\small $K$};
		\node at (1.55,-0.65) {\small $J$};
	\end{tikzpicture}
	\ \ =\ \ 
	\begin{tikzpicture}[xscale=-1,baseline=(current bounding box.center),scale=1.15,xscale=1.1]
		\draw[fill=gray!30] (0,0) rectangle (1,0.75);
		\draw[fill=gray!30] (0.7,1.25) rectangle (1.7,2);
		\node at (0.5,0.375) {\small $\mathcal{Y}_{J, B}$};
		\node at (1.2,1.625) {\small $\mathcal{X}_{I, A}$};
		\draw[thick,->,>=stealth] (0.15,-0.5) -- (0.15,0);
		\draw[thick,->,>=stealth] (0.85,-0.5) -- (0.85,0);
		\draw[thick,->,>=stealth] (0.15,0.75) -- (0.15,2.5);
		\draw[thick,->,>=stealth] (1.55,-0.5) -- (1.55,1.25);
		\draw[thick,->,>=stealth] (0.85,0.75) to node[right] {\small $K$} (0.85,1.25);
		\draw[thick,->,>=stealth] (0.85,2) to node[right] {\small $K$} (0.85,2.35);
		\draw[thick,->,>=stealth] (1.55,2) -- (1.55,2.5);
		\draw[very thick] (0.7,2.35) -- (1,2.35);
		\node at (0.15,2.65) {\small $B$};
		\node at (1.55,2.65) {\small $A$};
		\node at (0.15,-0.65) {\small $J$};
		\node at (0.85,-0.65) {\small $K$};
		\node at (1.55,-0.65) {\small $I$};
	\end{tikzpicture}
	\ \ =\ \ 
	\begin{tikzpicture}[baseline=(current bounding box.center),scale=1.15,xscale=1.1]
		\draw[fill=gray!30] (0.1,0) rectangle (1.4,0.75);
		\node at (0.75,0.375) {$\mathcal{D}$};
		\draw[thick,->,>=stealth] (0.75,-0.5) -- (0.75,0);
		\draw[thick,->,>=stealth] (-0.3,-0.5) -- (-0.3,1.25);
		\draw[thick,->,>=stealth] (1.8,-0.5) -- (1.8,1.25);
		\draw[thick,->,>=stealth] (0.35,0.75) to node[right=-0.05cm,pos=0.5] {\small $K$} (0.35,1.25);
		\draw[thick,->,>=stealth] (1.15,0.75) to node[right=-0.05cm,pos=0.5] {\small $K$} (1.15,1.25);
		\draw[fill=gray!30] (-0.45,1.25) rectangle (0.65,2);
		\draw[fill=gray!30] (0.85,1.25) rectangle (1.95,2);
		\node at (0.1,1.625) {\small $\mathcal{X}_{I, A}$};
		\node at (1.4,1.625) {\small $\mathcal{Y}_{J, B}$};
		\draw[thick,->,>=stealth] (-0.3,2) -- (-0.3,2.5);
		\draw[thick,->,>=stealth] (1.8,2) -- (1.8,2.5);
		\draw[thick,->,>=stealth] (0.3,2) to node[right=-0.05] {\small $K$} (0.3,2.35);
		\draw[thick,->,>=stealth] (1.2,2) to node[right=-0.05] {\small $K$} (1.2,2.35);
		\draw[very thick] (0.15,2.35) -- (0.45,2.35);
		\draw[very thick] (1.05,2.35) -- (1.35,2.35);
		\node at (-0.3,2.65) {\small $A$};
		\node at (1.8,2.65) {\small $B$};
		\node at (-0.3,-0.65) {\small $I$};
		\node at (0.75,-0.65) {\small $K$};
		\node at (1.8,-0.65) {\small $J$};
	\end{tikzpicture}
	\caption{\label{fig:simplifiedscenario} \textbf{Illustration of the scenario described in the introduction.} The circuit diagram corresponds to a special case of the diagram shown in \cref{fig:corollary} (see also \cref{cor:commute}). To see the correspondence, set $\mathcal{X}=\mixed{I} \circ \mathcal{X}_{I, A} \otimes \mathcal{I}_J$ and $\mathcal{Y}=\mathcal{I}_I \otimes \mathcal{Y}_{J, B} \circ \mixed{J}$, where $\mixed{I}$ and $\mixed{J}$ are completely positive maps that output a mixed state on $I$ and $J$, respectively. The first two diagrams then match the first two diagrams of \cref{fig:corollary}, where  $H = I \otimes K \otimes J$. Similarly, the diagram to the right matches the diagram to the right of \cref{fig:corollary} for $\overline{\mathcal{X} }= \tr_{K} \circ \mathcal{X}_{I, A}$, and $\overline{\mathcal{Y}} = \tr_{K} \circ \mathcal{Y}_{J, B}$.}
\end{figure}
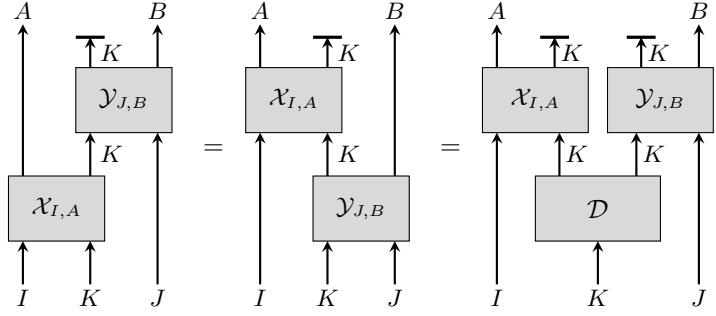

The paper is structured as follows: Section \ref{sec:prelims} outlines mathematical preliminaries.
In Section~\ref{sec:theorem}, we present the main theorem of this work, \cref{thm:tensorP}, which provides sufficient conditions for two maps such that their action on a system~$K$ can be factorised. We then prove that these conditions are also necessary (\cref{thm:converse}).  In Section~\ref{sec:corollaries}, we show how our main theorem answers the question posed above (\cref{cor:commute}). In the same section, we demonstrate that Tsirelson's answer to his question (phrased as \cref{cor:vNalgebras}) can be retrieved from \cref{cor:commute}. Finally, we present a generalisation of our main theorem to more than two parties operating on $K$ (\cref{cor:multimap}).

%% file: 1_preliminaries.tex
\section{Preliminaries}
\label{sec:prelims}

\newcommand{\Minf}[3]{I(#1\! :\!#2)_{#3}}

This section collects definitions, notation and theorems that are used in the proofs of the paper.

\begin{notation}
  We label Hilbert spaces with capital letters $H$, $K$, and so on. We also associate the same labels to the corresponding spaces of operators on these Hilbert spaces. We sometimes use the term systems to refer to these spaces. For example, we write $\rho_{H K}$ for a state (density operator) on the product of two systems~$H$ and~$K$. Furthermore, we use the notation $\rho_H \coloneqq \tr_{K}(\rho_{H K})$, where $\tr_K$ is the partial trace over~$K$. 
\end{notation}


\begin{notation}
	We denote by $\mathrm{id}_H$ the identity operator on~$H$. For $H$ finite-dimensional, $\mixed{H}$ is the normalised maximally mixed state on~$H$.
\end{notation}

\begin{notation}
	We write $\mathcal{M}:H\to K$ or $\mathcal{M}_{H \to K}$ to indicate that a completely positive (CP) map $\mathcal{M}$ goes from a system $H$ to a system $K$. That is, the map takes as input a trace-class operator on $H$ and outputs a trace-class operator on $K$. For a CP map $\mathcal{M}_{H\to K R}$ we use the notation $\mathcal{M}_{H\to K} \coloneqq\mathrm{tr}_R\circ \mathcal{M}_{H\to K R}$.	We usually omit identity maps, i.e., $\mathcal{M}_{H\to K}(\rho_{HR})\coloneqq \left(\mathcal{M}_{H\to K}\otimes\mathcal{I}_{R}\right)(\rho_{H R})$. 
\end{notation}

\begin{rem} \label{rem:TP}
  A CP map $\mathcal{M}: H \to K$ is trace-preserving (TP) if and only if $\tr(\mathcal{M}(W_H)) = \mathrm{tr}(W_H)$ holds for any trace-class operator $W_H$. This may also be written as 
  \begin{align} \label{eq:tracepreservingdef}
    \tr_K \circ \mathcal{M}_{H \to K} = \tr_H.
  \end{align}
  Note also that, if $\mathcal{M}$ has the Kraus representation $\mathcal{M}: \, W_H \mapsto \sum_z E_z W_H E_z^*$, then the TP property is equivalent to $\sum_z E_z^* E_z = \mathrm{id}_H$. 
  
  Similarly, $\mathcal{M}$ is trace non-increasing if and only if $\tr(\mathcal{M}(W_H)) \leq \mathrm{tr}(W_H)$ for any  trace-class operator $W_H \geq 0$ or, equivalently, if the Kraus operators satisfy $\sum_z E_z^* E_z \leq \mathrm{id}_H$. Note that this also implies the operator inequality 
  \begin{align}
    \tr_K \circ \mathcal{M}_{H \to K}(W_{H R})) \leq \tr_H(W_{H R}) \qquad \forall \, W_{H R} \geq 0.
  \end{align}
\end{rem}

\begin{notation} \label{not:stategeneratingmap}
   For any state $\rho_H$, we can define a CPTP map from $\mathbb{C}$ to $H$, which takes a trivial ($1$-dimenstional) system as input and outputs $\rho_H$, i.e., 
   \begin{equation}
     W \mapsto W \rho_H.
   \end{equation}
   We denote this map also by~$\rho_H$. Note that the concatenation $\tr_H \circ \rho_H$ is equal to the identity map. 
\end{notation}

\begin{defn}
	A CP map $\mathcal{M}:H \to H$ is unital if $\mathcal{M}(\mathrm{id}_H) = \mathrm{id}_H$. 
\end{defn}

\begin{defn}
	A CP map $\mathcal{M}: H \otimes I \to K$ is independent of $I$ if there exists a CP map $\overline{\mathcal{M}}: H\to K$ such that 
	\begin{equation} \label{eq:independence}
		\mathcal{M}_{HI\to K}=\overline{\mathcal{M}}_{H\to K}\circ \mathrm{tr}_I.
	\end{equation}
\end{defn}

\begin{rem} \label{rem:mapunique}
  If $\mathcal{M}: H \otimes I \to K$ is independent of $I$ then the map $\overline{\mathcal{M}}_{H\to K}$ in~\cref{eq:independence} is unique and equal to the map $\mathcal{M}_{H I \to K} \circ \zeta_I$, i.e., 
  \begin{equation}
     \overline{\mathcal{M}}_{H\to K}:W_H \mapsto \mathcal{M}_{HI \to K}(W_H \otimes \zeta_I),
  \end{equation}
  where $\zeta_I$ is an arbitrary state on~$I$.
\end{rem}

\begin{lem}
	\label{lem:fact1}
	For positive operators $\rho_{GH}$ and $\sigma_{KH}$ where $\rho_{GH}$ is pure and $\rho_H=\sigma_H$, there exists a CPTP map $\mathcal{R}: G \to K$ such that
		\begin{equation}
			\sigma_{KH}=\mathcal{R}_{G\to K}(\rho_{GH}).
		\end{equation}
\end{lem}

\begin{proof}
	Let $\tilde{\sigma}_{KEH}$ be a purification of $\sigma_{KH}$. The vector representations of the pure states $\rho_{GH}$ and $\tilde{\sigma}_{KEH}$ then have Schmidt decompositions $\sum_{i \in \mathfrak{I}} \lambda_i \ket{g_i}_{G} \otimes \ket{h_i}_H$ and $\sum_{i \in \mathfrak{I}} \lambda_i  \ket{e_i}_{KE} \otimes \ket{h_i}_H$, where $\{\ket{g_i}_{G}\}_{i \in \mathfrak{I}}$,  $\{\ket{e_i}_{KE}\}_{i \in \mathfrak{I}}$, and $\{\ket{h_i}_H\}_{i \in \mathfrak{I}}$ are orthonormal families of eigenvectors of  $\rho_G$, $\tilde{\sigma}_{KE}$, and $\rho_H = \tilde{\sigma}_H$, respectively. By adding additional orthonormal vectors, we can extend $\{\ket{g_i}_{G}\}_{i \in \mathfrak{I}}$ to an orthonormal basis $\{\ket{g_i}_{G}\}_{i \in \mathfrak{I'}}$ of $G$. And because we can without loss of generality choose $K$ such that the space $K E$ is larger than $G$, we can also add orthonormal vectors to $\{\ket{e_i}_{KE}\}_{i \in \mathfrak{I}}$ to obtain a larger family $\{\ket{e_i}_{K E}\}_{i \in \mathfrak{I'}}$. We may now define an isometry $V$ from $G$ to $K E$ by $\ket{g_i}_G \mapsto \ket{e_i}_{K E}$ for any $i \in \mathfrak{I'}$.  Then, $\mathcal{R}_{G \to K E} : W_G \mapsto V W_G V^{*}$ is a CPTP map with the property $\tilde{\sigma}_{K E H}=\mathcal{R}_{G \to K E} (\rho_{G H})$. We thus have $\sigma_{KH}=\mathrm{tr}_E \circ \mathcal{R}_{G\to K E}(\rho_{GH})$ as desired.
\end{proof}

\begin{rem}
	\label{rem:CJ}
	We will make heavy use of the Choi-Jamio\l kowski (C.-J.) isomorphism \cite{Jamiolkowski1972,Choi1975}, according to which a CP map $\mathcal{M}:H \to K$, where $H$ has finite dimension, can be represented as a bipartite positive operator $\rho_{K \tilde{H}} := \mathcal{M}(\psi_{ H \tilde{H}})$, where $\psi_{H\tilde{H}}\coloneqq\ketbra{\psi}{\psi}_{H\tilde{H}}$ is a maximally entangled state between $H$ and an isomorphic system~$\tilde{H}$. The C.-J.~isomorphism depends on the choice of $\smash{\psi_{H \tilde{H}}}$, which we will thus assume to be fixed. Note that, if $H$ is composed of subsystems, then $\psi_{H \tilde{H}}$ induces an analogous subsystem structure on $\tilde{H}$. 
	
	We will, in particular, consider spaces that decompose as $H = \bigoplus_z H_A^z \otimes H_B^z$. To reflect this decomposition on $\tilde{H}$, we equip $H$ with an orthonormal basis of the form $\{ \ket{a}_{H^z_A} \otimes \ket{b}_{H^z_B}\}_{z, a, b}$, where, for any $z$, $\{\ket{a}_{H^z_A}\}_{a \in \mathfrak{A}^z}$ and $\{\ket{b}_{H^z_B}\}_{b \in \mathfrak{B}^z}$ are orthonormal bases of $H^z_A$ and $H^z_B$, respectively, and write the Schmidt decomposition of $\ket{\psi}_{H \tilde{H}}$ as
	\begin{align}
	   \ket{\psi}_{H \tilde{H}} = \sqrt{{\textstyle \frac{1}{\mathrm{dim}(H)}}} \sum_{z} \sum_{\substack{a \in \mathfrak{A}^z \\ b \in \mathfrak{B}^z}} \bigl( \ket{a}_{H^z_A} \otimes \ket{b}_{H^z_B} \bigr) \otimes \ket{\varphi_{z, a, b}} _{\tilde{H}}, 
	\end{align}
        where $\ket{\varphi_{z, a, b}}_{\tilde{H}}$ are appropriately chosen normalised vectors on $\tilde{H}$. We may now, for any fixed $z$, define the subspace $\tilde{H}^z \coloneqq \mathrm{span} \{\ket{\varphi_{z, a, b}}\}_{a \in \mathfrak{A}^z, b \in \mathfrak{B}^z}$. Furthermore, we may introduce new spaces $\tilde{H}_A^z$ and $\tilde{H}_B^z$ with orthonormal bases $\{\ket{a}_{\tilde{H}^z_A}\}_{a \in \mathfrak{A}^z}$ and $\{\ket{b}_{\tilde{H}^z_B}\}_{b \in \mathfrak{B}^z}$, respectively, and define their tensor product by the bilinear map $\otimes: \, \tilde{H}_A^z \times \tilde{H}_B^z \to \tilde{H}^z$, which maps  $(\ket{a}_{\tilde{H}^z_A}, \ket{b}_{\tilde{H}^z_B})$ to $\ket{\varphi_{z, a, b}}$, for any $a \in \mathfrak{A}^z, b \in \mathfrak{B}^z$. This definition ensures that $\smash{\tilde{H}^z = \tilde{H}_A^z \otimes \tilde{H}_B^z}$. The maximally entangled state $\ket{\psi}_{H \tilde{H}}$ can then be expressed as 
        	\begin{align} \label{eq:psidecomposed}
	   \ket{\psi}_{H \tilde{H}} = \sqrt{{\textstyle \frac{1}{\mathrm{dim}(H)}}} \sum_{z} \sum_{\substack{a \in \mathfrak{A}^z \\ b \in \mathfrak{B}^z}} \bigl( \ket{a}_{H^z_A} \otimes \ket{b}_{H^z_B} \bigr) \otimes  \bigl( \ket{a}_{\tilde{H}^z_A} \otimes \ket{b}_{\tilde{H}^z_B} \bigr)
	\end{align}
	A special case of this is if $H$ factorises into $H_A \otimes H_B$. Then  there exists a factorisation of $\tilde{H}$ into $\tilde{H}_A \otimes \tilde{H}_B$ such that $\ket{\psi}_{H \tilde{H}} = \ket{\psi}_{H_A \tilde{H}_A} \otimes \ket{\psi}_{H_B \tilde{H}_B}$, where $\ket{\psi}_{H_A \tilde{H}_A}$ and $\ket{\psi}_{H_B \tilde{H}_B}$ are maximally entangled states on $H_A \otimes \tilde{H}_A$ and $H_B \otimes \tilde{H}_B$, respectively. 
		
	We summarise some further basic properties of the C.-J.~isomorphism (see \cref{app:CJ} for proofs): The map $\mathcal{M}$ is TP if and only if $\rho_{\tilde{H}}=\mixed{\tilde{H}}$. In this case $\rho_{K \tilde{H}}$ is normalised and thus a state. Furthermore, $\mathcal{M}$ is trace non-increasing if and only if $\rho_{\tilde{H}}\le\mixed{\tilde{H}}$. $\mathcal{M}$ can be retrieved from $\rho_{K\tilde{H}}$ via
	\begin{equation}
		\mathcal{M}(W_H)=\dim(H)^2\,\mathrm{tr}_{\tilde{H}}\Big(\mathrm{tr}_H\big(W_H\psi_{H\tilde{H}}\big)\rho_{K\tilde{H}}\Big).
	\end{equation}
\end{rem}

\begin{lem} \label{lem:unitaltrace}
	Let $\mathcal{M}:H\to H$ be a unital, trace non-increasing CP map on a finite-dimensional space~$H$. Then $\mathcal{M}$ is TP.
\end{lem}

\begin{proof}
	According to \cref{rem:CJ}, a map $\mathcal{M}:H\to H$ is trace non-increasing if and only if its C.-J.~operator $\rho_{H\tilde{H}}=\mathcal{M}(\psi_{H\tilde{H}})$ fulfils
		\begin{equation}
			\label{eq:tracenonincr}
			\rho_{\tilde{H}}=\mathrm{tr}_H\rho_{H\tilde{H}}\le\mixed{\tilde{H}}.
		\end{equation}
	Because $\mathcal{M}$ is unital, we also know that 
		\begin{equation}
			\rho_H=\mathrm{tr}_{\tilde{H}}\circ\mathcal{M}(\psi_{H\tilde{H}})=\mathcal{M}\left(\mixed{H}\right)=\mixed{H}.
		\end{equation}
	The latter implies $\mathrm{tr}(\rho_{\tilde{H}})=\mathrm{tr}(\rho_{H\tilde{H}})=\mathrm{tr}(\rho_H)=\mathrm{tr}(\mixed{H}) = \mathrm{tr}(\mixed{\tilde{H}})$. But this can only be true if the operator inequality \cref{eq:tracenonincr} is an equality, i.e., $\rho_{\tilde{H}}=\mixed{\tilde{H}}$. Hence, from the C.-J.~isomorphism (see again \cref{rem:CJ}) it follows that $\mathcal{M}$ is TP.
\end{proof}

\begin{lem}
	\label{lem:mut0}
	Let $\mathcal{M}:H \otimes I \to K$ be a CPTP map for finite-dimensional spaces $H, I, K$ such that $\mathcal{M}$ is independent of $I$, and let $\rho_{K\tilde{H}\tilde{I}}$ be the C.-J.~state of $\mathcal{M}$. Then
		\begin{equation}
			\Minf{K}{\tilde{I}|\tilde{H}}{\rho}=0.
		\end{equation}
\end{lem}

\begin{proof}
	Because $\mathcal{M}$ is independent of $I$, there exists a CPTP map $\overline{\mathcal{M}}:H\to K$ such that $\overline{\mathcal{M}}\circ\mathrm{tr}_I=\mathcal{M}$. Then, for a maximally entangled state $\psi_{HI\tilde{H}\tilde{I}}=\psi_{H\tilde{H}} \otimes \psi_{I\tilde{I}}$ (see \cref{rem:CJ}),
		\begin{align}
			\rho_{K\tilde{H}\tilde{I}}&=\left(\mathcal{M}\otimes\mathcal{I}_{\tilde{H}\tilde{I}}\right)(\psi_{HI\tilde{H}\tilde{I}})\\
			&=\left(\overline{\mathcal{M}}\circ\mathrm{tr}_I\otimes\mathcal{I}_{\tilde{H}\tilde{I}}\right)(\psi_{H\tilde{H}} \otimes \psi_{I\tilde{I}})\\
			&=\left(\overline{\mathcal{M}}\otimes\mathcal{I}_{\tilde{H}} \otimes \mathcal{I}_{\tilde{I}}\right)\left(\psi_{H\tilde{H}} \otimes \mixed{\tilde{I}}\right)\\
			&=\left(\overline{\mathcal{M}}\otimes\mathcal{I}_{\tilde{H}}\right)(\psi_{H\tilde{H}}) \otimes \mixed{\tilde{I}}.
		\end{align}
	From this tensor product structure of $\rho_{H\tilde{I}\tilde{R}}$, it follows that
		\begin{align}
			H(K|\tilde{H}\tilde{I})_\rho&=H(K\tilde{H}\tilde{I})_\rho-H(\tilde{H}\tilde{I})_\rho\\
			&=H(K\tilde{H})_\rho+H(\tilde{I})_\rho-H(\tilde{H})_\rho-H(\tilde{I})_\rho\\
			&=H(K|\tilde{H})_\rho,
		\end{align}
	hence $\Minf{K}{\tilde{I}|\tilde{H}}{\rho}=0$.
\end{proof}

\begin{lem}
	\label{lem:1dmaps}
	For any CP map $\mathcal{M}:H\to\mathbb{C}$ there exists a Hermitian operator $M_H$ such that
		\begin{equation}
			\mathcal{M}(W_H)=\mathrm{tr}(M_H W_H).
		\end{equation}
\end{lem}

\begin{proof}
	Let $\{E_z\}_z$ be the Kraus operators of $\mathcal{M}$, i.e., $\mathcal{M}: \, W_H \mapsto \sum_zE_z W_H E_z^*$. Because the image of $\mathcal{M}$ is one-dimensional, we have that $\mathcal{M}(W_H)=\mathrm{tr}(\mathcal{M}(W_H))$. Hence, using cyclicity and linearity of the trace,
		\begin{align}
			\mathcal{M}(W_H)&=\mathrm{tr}\left(\sum_zE_z W_H E_z^*\right)=\mathrm{tr}\left(\sum_zE_z^*E_z W_H\right)=\mathrm{tr}(M W_H),
		\end{align}
	where $M=\sum_zE_z^*E_z$.
\end{proof}

\begin{rem} \label{rem:cpmrescaling}
  Let $\mathcal{M}:H\to K$ be a CP map. If~$H$ is finite-dimensional then 
  \begin{align}
    \lambda \coloneqq \sup_{\rho_H} \tr(\mathcal{M}(\rho_H)),
  \end{align}
  where the supremum ranges over all states on~$H$, is finite. Hence, the rescaled map $\frac{1}{\lambda} \mathcal{M}$ is trace non-increasing. 
\end{rem}

\begin{lem} \label{lem:cpmtracepreserving}
  Let $\mathcal{M}:H\to K$ be a trace non-increasing CP map and let $K' \coloneqq K\oplus\mathrm{span}\{\ket{\perp}\}$, with $\ket{\perp}$ a unit vector. Then the map 
 from $H$ to $K'$ defined by\footnote{Here, $\perp_{K'}$ denotes the CPTP map that generates the state $\perp_{K'} = \ketbra{\perp}{\perp}_{K'}$; see~\cref{not:stategeneratingmap}.}
  \begin{align} \label{eq:tracepreservingmap}
     \mathcal{M}' \coloneqq \mathcal{M} + \perp_{K'} \circ \bigl(\mathrm{tr}_H  - \mathrm{tr}_{K} \circ \mathcal{M} \big),
  \end{align}
  is CP and TP. Furthermore, if $H= I \otimes J$ and $\mathcal{M}$ is independent of~$I$, then $\mathcal{M}'$ is also independent of~$I$. 
\end{lem}

\begin{proof}
  We start by showing the complete positivity of $\mathcal{M}'$. It suffices to verify that the map $\tr_H  - \tr_K \circ \mathcal{M}$ is CP, which is equivalent to 
  \begin{align}
    \tr_H(\rho_{H \tilde{H}}) \ge \tr_K(\mathcal{M}(\rho_{H \tilde{H}})) \quad \forall \rho_{H \tilde{H}} \ge 0.
  \end{align}
  This operator inequality follows from the assumption that $\mathcal{M}$ is trace non-increasing (see \cref{rem:TP}).
  The map $\mathcal{M}'$ is also TP. This can be verified by taking the trace on the right-hand side of~\cref{eq:tracepreservingmap} and noting that $\tr(\perp_{K'}) = 1$. 
  Finally, the independence of $\mathcal{M}'$ from $I$ can be verified by inspecting the right-hand side of~\cref{eq:tracepreservingmap}, where the maps $\mathcal{M}$ and $\mathrm{tr}_H$ are independent of~$I$.
\end{proof}

For the convenience of the reader we also state here Theorem 6 of~\cite{Hayden2004}, since this will be used in later proofs.

\begin{thm}[Theorem 6 of~\cite{Hayden2004}]
	\label{thm:Hayden}
	Let $A,B,C$ be finite-dimensional Hilbert spaces. A state $\rho_{ABC}$ on $A\otimes B\otimes C$ satisfies $\Minf{A}{B|C}{\rho}=0$ if and only if there is a decomposition of~$C$ as
	\begin{equation}
		C=\bigoplus_z C_A^z\otimes C_B^z
	\end{equation}
	into a direct sum of tensor products, such that
	\begin{equation}
		\rho_{ABC}=\sum_z p_z\,\rho^{z}_{AC_A^z}\otimes\rho^{z}_{C_B^z B}
	\end{equation}
	with states $\rho^{z}_{AC_A^z}$ on $A\otimes C_A^z$ and $\rho^{z}_{C_B^z B}$ on $C_B^z\otimes B$, and a probability distribution $\{p_z\}$. 
\end{thm}

%% file: 2_theorem.tex
\section{Main result}
\label{sec:theorem}

The claim that commuting operations factorise, as described informally in the introduction, is a corollary from a more general statement, \cref{thm:tensorP}, which we present and prove in the following. The setting considered by the theorem is illustrated by~\cref{fig:theorem}. At the end of the section we also give a converse statement, \cref{thm:converse}, which implies that the assumptions we make in \cref{thm:tensorP} are necessary for factorisation.
\begin{figure}[t]
	\centering
	\begin{tikzpicture}[baseline=(current bounding box.center),scale=1.25]
		\draw[fill=gray!30] (0,0) rectangle (1,0.75);
		\draw[fill=gray!30] (0.5,1.25) rectangle (1.5,2);
		\node at (0.5,0.375) {\small $\mathcal{M}$};
		\node at (1,1.625) {\small $\mathcal{N}$};
		\draw[thick,->,>=stealth] (0.5,-0.5) -- (0.5,0);
		\draw[thick,->,>=stealth] (0.25,0.75) to node[left] {\small $A$} (0.25,2.35);
		\draw[thick,->,>=stealth] (0.75,0.75) to node[right] {\small $H$} (0.75,1.25);
		\draw[thick,->,>=stealth] (1,2) -- (1,2.5);
		\draw[very thick] (0.1,2.35) -- (0.4,2.35);
		\node at (0.5,-0.65) {\small $H$};
		\node at (1,2.65) {\small $B$};
	\end{tikzpicture}\ \ =\ \ 
	\begin{tikzpicture}[baseline=(current bounding box.center),scale=1.1]
		\draw[fill=gray!30] (0,0) rectangle (1,0.75);
		\node at (0.5,0.375) {\small $\mathcal{N}$};
		\draw[thick,->,>=stealth] (0.5,-0.5) -- (0.5,0);
		\node at (0.5,-0.65) {\small $H$};
		\draw[thick,->,>=stealth] (0.5,0.75) -- (0.5,1.25);
		\node at (0.5,1.4) {\small $B$};
	\end{tikzpicture}\hspace{15pt}$\Rightarrow$\hspace{15pt}
	\begin{tikzpicture}[baseline=(current bounding box.center),scale=1.25]
		\draw[fill=gray!30] (0,0) rectangle (1,0.75);
		\draw[fill=gray!30] (0.5,1.25) rectangle (1.5,2);
		\node at (0.5,0.375) {\small $\mathcal{M}$};
		\node at (1,1.625) {\small $\mathcal{N}$};
		\draw[thick,->,>=stealth] (0.5,-0.5) -- (0.5,0);
		\draw[thick,->,>=stealth] (0.25,0.75) -- (0.25,2.5);
		\draw[thick,->,>=stealth] (0.75,0.75) to node[right] {\small $H$} (0.75,1.25);
		\draw[thick,->,>=stealth] (1,2) -- (1,2.5);
		\node at (0.5,-0.65) {\small $H$};
		\node at (1,2.65) {\small $B$};
		\node at (0.25,2.65) {\small $A$};
	\end{tikzpicture}\ \ =\ \ 
	\begin{tikzpicture}[baseline=(current bounding box.center),scale=1.25]
		\draw[fill=gray!30] (0.1,0) rectangle (1.4,0.75);
		\node at (0.75,0.375) {$\mathcal{D}$};
		\draw[thick,->,>=stealth] (0.75,-0.5) -- (0.75,0);
		\draw[thick,->,>=stealth] (-0.3,-0.5) -- (-0.3,1.25);
		\draw[thick,->,>=stealth] (1.8,-0.5) -- (1.8,1.25);
		\draw[thick,->,>=stealth] (0.35,0.75) to node[right=-0.05cm,pos=0.5] {\small $K$} (0.35,1.25);
		\draw[thick,->,>=stealth] (1.15,0.75) to node[right=-0.05cm,pos=0.5] {\small $K$} (1.15,1.25);
		\draw[fill=gray!30] (-0.45,1.25) rectangle (0.65,2);
		\draw[fill=gray!30] (0.85,1.25) rectangle (1.95,2);
		\node at (0.1,1.625) {$\overline{\mathcal{M}}$};
		\node at (1.4,1.625) {$\overline{\mathcal{N}}$};
		\draw[thick,->,>=stealth] (0.1,2) -- (0.1,2.5);
		\draw[thick,->,>=stealth] (1.4,2) -- (1.4,2.5);
		\node at (0.1,2.65) {\small $A$};
		\node at (1.4,2.65) {\small $B$};
		\node at (-0.3,-0.65) {\small $I$};
		\node at (0.75,-0.65) {\small $K$};
		\node at (1.8,-0.65) {\small $J$};
		\draw[very thick, decorate,decoration={calligraphic brace,mirror}] (-0.35,-0.85) -- (1.85,-0.85);
		\node at (0.75,-1.15) {\small $=H$};
	\end{tikzpicture}
	\caption{\label{fig:theorem}\textbf{Visualisation of \cref{thm:tensorP}.} The equality on the left-hand side illustrates Condition~\ref{cond:1}. Given that Conditions~\ref{cond:2} and~\ref{cond:3} are also satisfied, the theorem implies the equality between the circuit diagrams on the right-hand side.}
\end{figure}
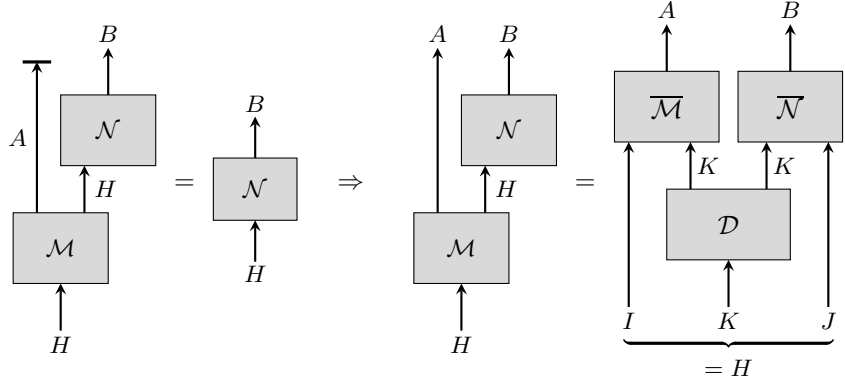

\begin{thm}
	\label{thm:tensorP}
	Let $\mathcal{M}:H\to A\otimes H$ and $\mathcal{N}:H\to B$ be CP maps, where $H= I \otimes K\otimes J$ is  finite-dimensional, such that
		\begin{enumerate}[label=(\roman*)]
			\item \label{cond:1} $\mathrm{tr}_A\circ\mathcal{N}\circ\mathcal{M}=\mathcal{N}$
			\item \label{cond:2} $\mathrm{tr}_A\circ\mathcal{M}$ is unital and trace non-increasing
			\item \label{cond:3} $\mathrm{tr}_H\circ\mathcal{M}$ is independent of $J$ and $\mathcal{N}$ is independent of $I$. 
		\end{enumerate}
	Then there exists a completely positive, trace-preserving map $\mathcal{D}:K\to K\otimes K$ (``doubling map'') such that
		\begin{align} \label{eq:mainresult}
			\mathcal{N}\circ\mathcal{M}=\big(\overline{\mathcal{M}}\otimes\overline{\mathcal{N}}\big)\circ\mathcal{D},
		\end{align}
	where $\overline{\mathcal{M}}\circ\tr_J=\tr_H\circ\mathcal{M}$, $\overline{\mathcal{N}}\circ\tr_I=\mathcal{N}$.\footnote{The maps $\overline{\mathcal{M}}$ and $\overline{\mathcal{N}}$ are well-defined and unique; see \cref{rem:mapunique}.}
\end{thm}

\noindent
\textbf{Proof outline.} The proof of \cref{thm:tensorP} consists of the following steps:
	\begin{enumerate}[label=(\roman*)]
		\item Let $\rho_{AB\tilde{H}}$ be the C.-J.~operator of the map $\mathcal{N}\circ\mathcal{M}$, where $\tilde{H}$ is a Hilbert space isomorphic to $H$. Show that $\Minf{A\tilde{I}}{B \tilde{J}|\tilde{K}}{\rho}=0$, i.e., $H(B\tilde{J}|\tilde{K})_{\rho}=H(B\tilde{J}|A\tilde{I}\tilde{K})_{\rho}$, via the data-processing inequality (\textrightarrow\,\cref{claim:mutualinfomodified}).
		\item Apply \cref{thm:Hayden}, which yields that  $\rho_{AB\tilde{H}}$ is of the form
			\begin{equation}
				\label{eq:statedecomp}
				\rho_{AB\tilde{H}}=\sum_z p_z\,\rho^{z}_{A\tilde{I}\tilde{K}_A^z}\otimes\rho^{z}_{\tilde{K}_B^z\tilde{J} B}
			\end{equation}
			for a probability distribution $\{p_z\}$.
		\item Show that $\rho_{AB\tilde{H}}$ above is equal to the C.-J.~operator of the map $\big(\overline{\mathcal{M}}\otimes\overline{\mathcal{N}}\big)\circ\mathcal{D}$ (\textrightarrow\,\cref{claim:tensormap}). 	\end{enumerate}

\begin{proof}[Proof of \cref{thm:tensorP}]

	We give the proof here under the assumption that the CP maps $\mathcal{M}$ and $\mathcal{N}$ are TP and that $A$ and $B$ are finite-dimensional. As we will explain in \cref{rem:prooftracenoninc} and \cref{rem:infiniteAB} below, these assumptions can be made without loss of generality.

	First, we define a couple of quantum states that will be essential throughout the proof. Consider a Hilbert space $\tilde{H}$ that is isomorphic to $H$. The C.-J.~operator of $\mathcal{N}\circ \mathcal{M}$ is given by
		\begin{equation}
				\rho_{AB\tilde{H}}=\mathcal{N}\circ\mathcal{M}(\psi_{H \tilde{H}})\label{eq:defrho},
		\end{equation}
	where $\psi_{H\tilde{H}}\coloneqq\ketbra{\psi}{\psi}_{H\tilde{H}}$ is a maximally entangled state (see \cref{rem:CJ}). Thus, $\psi_{H}=\mixed{H}$ and $\psi_{\tilde{H}}=\mixed{\tilde{H}}$.
	Note that since we assume that $\mathcal{M}$ and $\mathcal{N}$ are TP, $\rho_{AB\tilde{H}}$ is normalised and hence a state.

	We will furthermore need the following quantum states (see \cref{fig:circuitdiagrams} for an illustration):
	\begin{align}
		\sigma_{AH\tilde{H}}&=\mathcal{M}(\psi_{H \tilde{H}})\label{eq:defsigma}\\
		\rho'_{B\tilde{H}}&=\mathcal{N}(\psi_{H\tilde{H}})\label{eq:defrhop}.
	\end{align}

		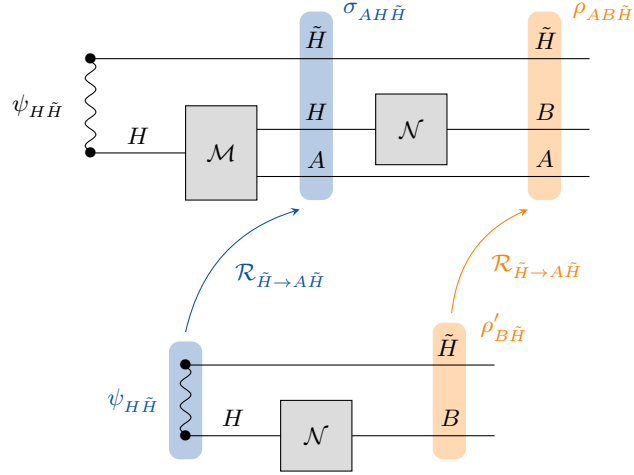
\begin{figure}[t]
		\centering
		\begin{tikzpicture}[scale=1.25]
			\draw[fill=myblue,draw=none,opacity=0.3,rounded corners] (2.2,-0.5) rectangle (2.55,1.5);
			\node[color=myblue] at (3,1.5) {\small $\sigma_{AH\tilde{H}}$};
			\draw[fill=myorange,draw=none,opacity=0.3,rounded corners] (4.6,-0.5) rectangle (4.95,1.5);
			\node[color=myorange] at (5.4,1.5) {\small $\rho_{AB\tilde{H}}$};
			\node at (0,0) {$\bullet$};
			\node at (0,1) {$\bullet$};
			\draw[decoration={snake,segment length=3mm, amplitude=.7mm},decorate] (0,0) -- (0,1);
			\node at (-0.55,0.5) {\small $\psi_{H\tilde{H}}$};
			\draw[fill=LightGray] (1,-0.5) rectangle (1.75,0.5);
			\node at (1.375,0) {\small $\mathcal{M}$};
			\draw[fill=LightGray] (3,-0.125) rectangle (3.75,0.625);
			\node at (3.375,0.25) {\small $\mathcal{N}$};
			\draw (0,0) to node[above] {\small $H$} (1,0);
			\draw (1.75,0.25) to node[above] {\small $H$} (3,0.25);
			\draw (3.75,0.25) to node[above,pos=0.69] {\small $B$} (5.25,0.25);
			\draw (1.75,-0.25) to node[above,pos=0.865] {\small $A$} (5.25,-0.25);
			\draw (0,1) to node[above,pos=0.912] {\small $\tilde{H}$} (5.25,1);
			\node at (2.375,1.225) {\small $\tilde{H}$};
			\node at (2.375,-0.075) {\small $A$};
			\begin{scope}[yshift=-3cm,xshift=1cm]
				\draw[fill=myblue,draw=none,opacity=0.3,rounded corners] (-0.175,-0.25) rectangle (0.175,1);
				\draw[fill=myorange,draw=none,opacity=0.3,rounded corners] (2.6,-0.25) rectangle (2.95,1.2);
				\node[color=myorange] at (3.35,1.15) {\small $\rho'_{B\tilde{H}}$};
				\node at (0,0) {$\bullet$};
				\node at (0,0.75) {$\bullet$};
				\draw[decoration={snake,segment length=3mm, amplitude=.7mm},decorate] (0,0) -- (0,0.75);
				\node[color=myblue] at (-0.55,0.375) {\small $\psi_{H\tilde{H}}$};
				\draw[fill=LightGray] (1,-0.375) rectangle (1.75,0.375);
				\node at (1.375,0) {\small $\mathcal{N}$};
				\draw (0,0) to node[above] {\small $H$} (1,0);
				\draw (0,0.75) to node[above,pos=0.85] {\small $\tilde{H}$} (3.25,0.75);
				\draw (1.75,0) to node[above,pos=0.69] {\small $B$} (3.25,0);
			\end{scope}
			\draw[color=myblue,->,>=stealth] (1,-1.9) to [bend left] (2.2,-0.6);
			\draw[color=myorange,->,>=stealth] (3.8,-1.7) to [bend left] (4.6,-0.6);
			\node[color=myblue] at (2,-1.3) {\small $\mathcal{R}_{\tilde{H}\to A\tilde{H}}$};
			\node[color=myorange] at (4.7,-1.2) {\small $\mathcal{R}_{\tilde{H}\to A\tilde{H}}$};
		\end{tikzpicture}
		\caption{\label{fig:circuitdiagrams} \textbf{Relations between states used in the proof of  \cref{thm:tensorP}.} The diagram shows the states defined in \cref{eq:defrho,eq:defsigma,eq:defrhop} and the CP maps that connect them.}
	\end{figure}

	\begin{customthm}{1}
		\label{claim:mutualinfomodified}
		For the C.-J. operator $\rho_{AB\tilde{H}}$ of $\mathcal{N}\circ\mathcal{M}$ defined in \cref{eq:defrho}, it holds that 
		\begin{equation}
			\Minf{A\tilde{I}}{B \tilde{J} |\tilde{K}}{\rho}=0.
		\end{equation}
	\end{customthm}
	\begin{proof}[Proof of \cref{claim:mutualinfomodified}]
		To prove the statement, we need to show that $H(B\tilde{J}|\tilde{K})_{\rho}=H(B\tilde{J}|A\tilde{I}\tilde{K})_{\rho}$.
		From strong subadditivity, it follows that 
		\begin{equation}
			H(B\tilde{J}|A \tilde{I} \tilde{K})_{\rho}\le H(B\tilde{J}|\tilde{K})_{\rho}.
		\end{equation}
	
		To show the other direction, note that from the unitality of $\tr_A\circ\mathcal{M}$ stated in Condition~\ref{cond:2}, we have that
		\begin{align} 
			\begin{split}
				\sigma_H 
				&=\mathrm{tr}_{A\tilde{H}}\circ\mathcal{M}(\psi_{H\tilde{H}})\\
				&=\mathrm{tr}_{A}\circ\mathcal{M}\left(\mixed{H}\right)\\
				&=\mixed{H}\\
				&=\psi_H,
			\end{split}
		\end{align}
		i.e., $\psi_{H\tilde{H}}$ is a purification of $\sigma_H$. From \cref{lem:fact1} with the assignment $G \to \tilde{H}$, $K \to A \otimes \tilde{H}$, $H \to H$, $\rho_{G H} \to \psi_{\tilde{H} H}$, and $\sigma_{K H} \to \sigma_{A \tilde{H} H}$, we know there exists a CPTP map $\mathcal{R}_{\tilde{H}\to A\tilde{H}}$ such that the state in \cref{eq:defsigma} can be written as
		\begin{equation} 
			\label{eq:uselem1}
			\sigma_{AH\tilde{H}}=\mathcal{R}_{\tilde{H}\to A\tilde{H}}(\psi_{H\tilde{H}}).
		\end{equation} 
		 From \cref{eq:uselem1}, it follows that
			\begin{equation}
				\label{eq:useclaim1}
				\rho_{AB\tilde{H}}=\mathcal{N}(\sigma_{AH\tilde{H}})=\mathcal{N}\circ\mathcal{R}_{\tilde{H}\to A\tilde{H}}(\psi_{H\tilde{H}})=\mathcal{R}_{\tilde{H}\to A\tilde{H}}\circ\mathcal{N}(\psi_{H\tilde{H}})=\mathcal{R}_{\tilde{H}\to A\tilde{H}}(\rho'_{B\tilde{H}}), 
			\end{equation} 
			where we have used that $\mathcal{N}$ and $\mathcal{R}$ act on different systems and thus commute  (see \cref{fig:circuitdiagrams}).
		With the chain rule for conditional entropy it follows that
		\begin{align}
			H(B\tilde{J}|\tilde{K})_{\rho'}&=H(B|\tilde{K}\tilde{J})_{\rho'}+H(\tilde{J}|\tilde{K})_{\rho'}\\
			&=H(B|\tilde{K}\tilde{J}\tilde{I})_{\rho'}+H(\tilde{J}|\tilde{K}\tilde{I})_{\rho'}\label{eq:intermediateH}\\
			&\le H(B|\tilde{K}\tilde{J}A\tilde{I})_{\rho}+H(\tilde{J}|\tilde{K}\tilde{I})_{\rho'}\\
			&=H(B|\tilde{K}\tilde{J}A\tilde{I})_{\rho}+H(\tilde{J}|\tilde{K}\tilde{I})_{\rho}\\
			&=H(B|\tilde{K}\tilde{J}A\tilde{I})_{\rho}+H(\tilde{J}|\tilde{K}A\tilde{I})_{\rho}+\Minf{A}{\tilde{J}|\tilde{K}\tilde{I}}{\rho}\\
			&=H(B\tilde{J}|\tilde{K}A\tilde{I})_{\rho}.
		\end{align}
		In the second line, because the map $\mathcal{N}$ is such that $B$ is independent of $I$, we can apply \cref{lem:mut0}, which yields $H(B|\tilde{K}\tilde{J})_{\rho'}=H(B|\tilde{K}\tilde{J}\tilde{I})_{\rho'}$. Also, we have used that $\rho'_{\tilde{H}}=\mixed{\tilde{H}}$ implies $H(\tilde{J}|\tilde{K})_{\rho'}=H(\tilde{J}|\tilde{K}\tilde{I})_{\rho'}$.
		In the third line we have used \cref{eq:useclaim1} together with the data processing inequality, and in the fourth line we have used that $\rho_{\tilde{H}}=\mixed{\tilde{H}}=\rho'_{\tilde{H}}$. The expression in the fifth line directly follows from the definition of the mutual information. The mutual information then vanishes because the map $\mathrm{tr}_H \circ \mathcal{M}$ is such that $A$ is independent of $J$, which allows us to apply \cref{lem:mut0}.
		Finally, because of Condition~\ref{cond:1},  $\rho'_{B\tilde{H}}=\mathcal{N}(\psi_{H\tilde{H}})=\mathrm{tr}_A\circ\mathcal{N}\circ\mathcal{M}(\psi_{H\tilde{H}})=\rho_{B\tilde{H}}$, hence it follows that 
			\begin{equation}
				H(B\tilde{J}|\tilde{K})_{\rho'}=H(B\tilde{J}|\tilde{K})_{\rho}.
			\end{equation}

		Summarising, we have shown that 
		\begin{equation}
			H(B\tilde{J}|\tilde{K}A \tilde{I})_{\rho}\le H(B\tilde{J}|\tilde{K})_{\rho}\le H(B\tilde{J}|\tilde{K}A \tilde{I} )_{\rho},
		\end{equation}
		and therefore 
		\begin{equation}
			H(B\tilde{J}|\tilde{K}A \tilde{I})_{\rho}= H(B\tilde{J}|\tilde{K})_{\rho}.
		\end{equation}
		Hence, $\Minf{B\tilde{J}}{A \tilde{I}|\tilde{K}}{\rho}=0$, which is the statement of the claim.
	\end{proof}

	Using \cref{thm:Hayden}, \cref{claim:mutualinfomodified} implies that there exists a decomposition of $\tilde{K}$ of the form
		\begin{equation}
		\label{eq:Ptildecomp}
			\tilde{K}=\bigoplus_z \tilde{K}_A^z\otimes\tilde{K}_B^z
		\end{equation}
	such that 
		\begin{equation}
			\label{eq:defrho2}
			\rho_{AB\tilde{H}}=\sum_z p_z\,\rho_{A\tilde{I}\tilde{K}_A^z}^{z}\otimes\rho_{\tilde{K}_B^z\tilde{J} B}^{z}.
		\end{equation} 
	Using the decomposition of $\tilde{K}$, we may decompose $\tilde{H}$ as 
		\begin{align}
			\tilde{H}&=\tilde{I}\otimes\tilde{K}\otimes\tilde{J}
			=\bigoplus_z\tilde{I}\otimes\tilde{K}_A^z\otimes\tilde{K}_B^z\otimes\tilde{J}
			=\bigoplus_z \tilde{A}^z\otimes \tilde{B}^z
			=\bigoplus_z \tilde{H}^z,
		\end{align}
	where we have introduced the notation $\tilde{A}^z\coloneqq\tilde{I}\otimes\tilde{K}_A^z$, $\tilde{B}^z\coloneqq\tilde{K}_B^z\otimes\tilde{J}$, and  $\tilde{H}^z\coloneqq\tilde{A}^z\otimes\tilde{B}^z$. We can thus rewrite \cref{eq:defrho2} as 
	\begin{align} \label{eq:defrho4}
	  \rho_{AB\tilde{H}} = \sum_z p_z\,\rho_{A\tilde{A}^z}^{z}\otimes\rho_{\tilde{B}^z B}^{z} \ .
	\end{align}
	Taking the trace over $A$ and $B$, and applying a projection $\Pi_{\tilde{H}^z}$ onto $\tilde{H}^z$, for any $z$, we have
	\begin{align} \label{eq:rhoOnHz}
	   p_z\,\rho_{\tilde{A}^z}^{z}\otimes\rho_{\tilde{B}^z}^{z}  = \Pi_{\tilde{H}^z}(\rho_{\tilde{H}}) 
	   = \Pi_{\tilde{H}^z}(\psi_{\tilde{H}})  
	   = \Pi_{\tilde{H}^z}(\mixed{\tilde{H}}) 
	   \sim \mathrm{id}_{\tilde{H}^z}
	   = \mathrm{id}_{\tilde{A}^z} \otimes \mathrm{id}_{\tilde{B}^z},
	\end{align}
        where the second equality follows from \cref{eq:defrho} and the TP property of $\mathcal{M}$ and $\mathcal{N}$, and the third from the fact that $\psi_{H \tilde{H}}$ is maximally entangled. To proceed, it will be convenient to introduce rescaled operators  
        	\begin{align}
		\label{eq:tauproperties}
		\tau_{A\tilde{A}^z}^{z} \sim \rho_{A\tilde{A}^z}^{z}  \qquad \text{and} \qquad
		\tau_{\tilde{B}^z B}^{z} \sim \rho_{\tilde{B}^z B}^{z},
	\end{align}	
        which are normalised such that  
        \begin{align} \label{eq:taunormalisation}
          \mathrm{tr}\bigl( \tau_{A\tilde{A}^z}^{z}  \bigr) = \dim(\tilde{A}^z) \qquad \text{and} \qquad  \mathrm{tr} \bigl( \tau_{\tilde{B}^z B}^{z} \bigr) = \dim(\tilde{B}^z).
        \end{align} 
        It then follows from \cref{eq:rhoOnHz} that
	\begin{align}
		\label{eq:tauproperties}
		\tau_{\tilde{A}^z}^{z} = \mathrm{id}_{\tilde{A}^z} \qquad \text{and} \qquad  
		\tau_{\tilde{B}^z}^{z} = \mathrm{id}_{\tilde{B}^z}.
	\end{align}	
With these operators, we may rewrite \cref{eq:defrho4} as
	\begin{equation}
		\label{eq:rhowithtausq}
		\rho_{AB\tilde{H}}=\sum_z q_z \tau_{A\tilde{A}^z}^{z}\otimes\tau_{\tilde{B}^z B}^{z}
	\end{equation}
	for some appropriately chosen weights $q_z$, which we will now determine. For this we again take the trace over $A$ and $B$ on both sides and apply the projection $\Pi_{\tilde{H}^z}$, which yields
		\begin{align}
	           \Pi_{\tilde{H}^z}(\rho_{\tilde{H}}) 
	           = q_z\,\tau^{z}_{\tilde{A}^z}\otimes \tau^{z}_{\tilde{B}^z}
	           = q_z\,\mathrm{id}_{\tilde{A}^z}\otimes\mathrm{id}_{\tilde{B}^z}
	           = q_z\,\mathrm{id}_{\tilde{H}^z}.
		\end{align}
	Since, according to \cref{eq:rhoOnHz}, this must also equal $\Pi_{\tilde{H}^z}(\mixed{\tilde{H}})$, we find
	\begin{align}
	   q_z\,\mathrm{id}_{\tilde{H}^z} 
	   = \Pi_{\tilde{H}^z}(\mixed{\tilde{H}})
	   = {\textstyle \frac{1}{\dim(H)}} \Pi_{\tilde{H}^z}(\mathrm{id}_{\tilde{H}}) 
	   = {\textstyle \frac{1}{\dim(H)}} \mathrm{id}_{\tilde{H}^z},
	\end{align} 
	which implies $q_z=\frac{1}{\dim(H)}$. Inserting this into \cref{eq:rhowithtausq}, we conclude that
	\begin{align} \label{eq:rhowithtaus}
	  \rho_{AB\tilde{H}} = {\textstyle \frac{1}{\dim(H)}} \sum_z  \tau_{A\tilde{A}^z}^{z}\otimes\tau_{\tilde{B}^z B}^{z}.
	\end{align}

	\begin{customthm}{2}
		\label{claim:tensormap}
		There exists a CPTP map $\mathcal{D}:K\to K\otimes K$ such that
			\begin{equation}
				\mathcal{N}\circ\mathcal{M}=\big(\overline{\mathcal{M}}\otimes\overline{\mathcal{N}}\big)\circ\mathcal{D},
			\end{equation}
		where $\overline{\mathcal{M}}\circ\tr_J=\tr_H\circ\mathcal{M}$, $\overline{\mathcal{N}}\circ\tr_I=\mathcal{N}$.
	\end{customthm}

	\begin{proof}[Proof of \cref{claim:tensormap}]
		In the following, we use the notation $\mathcal{D}:K\to K'\otimes K''$, where $K'=K=K''$, to make clear how the involved maps are acting on the different Hilbert spaces. Let
			\begin{equation}
				\label{eq:defD}
				\mathcal{D}(W_K)\coloneqq\sum_z V^{(z)}W_K {V^{(z)}}^*\otimes\mixed{K_B'^z}\otimes\mixed{K_A''^z},
			\end{equation}
		where 
			\begin{equation}
				\label{eq:defV}
				V^{(z)}\coloneqq\sum_{a,b}\left(\ket{a}_{K_A'^z}\otimes\ket{b}_{K_B''^z}\right)\left(\bra{a}_{K_A^z}\otimes\bra{b}_{K_B^z}\right).
			\end{equation}
		The map $\mathcal{D}$ is CP because each term in its definition is CP, and we will verify at the end of the proof that it is also TP. 
		
		Next, we calculate the C.-J.~operator $\xi_{AB\tilde{H}}$ of the map $\big(\overline{\mathcal{M}}\otimes\overline{\mathcal{N}}\big)\circ\mathcal{D}$ with respect to the same state $\psi_{H \tilde{H}} = \ketbra{\psi}{\psi}_{H \tilde{H}}$ as in \cref{eq:defrho}. Because, according to \cref{rem:CJ}, this state can be written as $\ket{\psi}_{H \tilde{H}} = \ket{\psi}_{I \tilde{I}} \otimes \ket{\psi}_{K \tilde{K}} \otimes \ket{\psi}_{J \tilde{J}}$, where
		\begin{align} \label{eq:psidecomposed2}
		\ket{\psi}_{K\tilde{K}}={\textstyle \sqrt{\frac{1}{\dim(K)}}} \sum_{z,a,b} \ket{a}_{K_A^z}\otimes \ket{b}_{K_B^z} \otimes \ket{a}_{\tilde{K}_A^z} \otimes \ket{b}_{\tilde{K}_B^z},
		\end{align}
		we find
			\begin{align}
				\xi_{AB\tilde{H}}&\coloneqq \big(\overline{\mathcal{M}}\otimes\overline{\mathcal{N}}\big)\circ\mathcal{D}\left(\psi_{H\tilde{H}}\right)\\
				&=\big(\overline{\mathcal{M}}\otimes\overline{\mathcal{N}}\big)\circ\mathcal{D}\left(\ketbra{\psi}{\psi}_{I\tilde{I}}\otimes \ketbra{\psi}{\psi}_{K\tilde{K}}\otimes \ketbra{\psi}{\psi}_{J\tilde{J}}\right)\\
				&=\big(\overline{\mathcal{M}}\otimes\overline{\mathcal{N}}\big)\Bigg({\textstyle \frac{1}{\dim(K)}} \sum_z\sum_{\substack{a,b,\\\bar{a},\bar{b}}}\Bigl(\ketbra{a}{\bar{a}}_{K_A'^z} \otimes \ketbra{b}{\bar{b}}_{K_B''^z}\Bigr)\otimes \Bigl(\ketbra{a}{\bar{a}}_{\tilde{K}_A^z} \otimes \ketbra{b}{\bar{b}}_{\tilde{K}_B^z}\Bigr)\\
				&\hspace{20pt}\otimes \mixed{K_B'^z}\otimes\mixed{K_A''^z}\otimes\ketbra{\psi}{\psi}_{I\tilde{I}}\otimes\ketbra{\psi}{\psi}_{J\tilde{J}}\Bigg)\\
				&= {\textstyle \frac{1}{\dim(K)}} \sum_z\xi_{A\tilde{A}^z }^{z}\otimes\xi_{B\tilde{B}^z}^{z},\label{eq:CJsigma}
			\end{align}
		where 
			\begin{align}
				\xi_{A\tilde{A}^z}^{z} 
				= \xi_{A\tilde{K}_A^z\tilde{I}}^{z}&\coloneqq\overline{\mathcal{M}}\left(\sum_{a,\bar{a}}\ketbra{a}{\bar{a}}_{K_A'^z}\otimes\mixed{K_B'^z}\otimes\ketbra{a}{\bar{a}}_{\tilde{K}_A^z} \otimes\ketbra{\psi}{\psi}_{I\tilde{I}} \right)\\
				\xi_{B\tilde{B}^z}^{z}
				= \xi_{B\tilde{K}_B^z\tilde{J}}^{z}&\coloneqq\overline{\mathcal{N}}\left(\sum_{b,\bar{b}}\mixed{K_A''^z}\otimes\ketbra{b}{\bar{b}}_{K_B''^z}\otimes\ketbra{b}{\bar{b}}_{\tilde{K}_B^z} \otimes\ketbra{\psi}{\psi}_{J\tilde{J}} \right).
			\end{align}
		It remains to show that $\xi_{AB\tilde{H}}$ equals the C.-J.~state $\rho_{AB\tilde{H}}$  of the map $\mathcal{N} \circ \mathcal{M}$. To this aim, we note that \cref{eq:taunormalisation} implies
			\begin{equation}
				\label{eq:tracedrho}
				\tr_{B\tilde{B}^z} \circ \Pi_{\tilde{H}^z} (\rho_{AB\tilde{H}})
				= {\textstyle \frac{\dim(\tilde{B}^z)}{\dim(H)}} \, \tau_{A\tilde{A}^z}={\textstyle \frac{\dim(K_B^z)}{\dim(K)\dim(I)}} \, \tau_{A\tilde{A}^z}.
			\end{equation}
		Furthermore, because $\mathcal{N}$ is TP, we have $\tr_B\circ\mathcal{N}\circ\mathcal{M}=\tr_H\circ \mathcal{M}$ (see \cref{rem:TP}). It thus follows that the partial trace $\rho_{A\tilde{H}} = \mathrm{tr}_B(\rho_{A B \tilde{H}})$ is the C.-J.~state of $\tr_H\circ\mathcal{M}$, and thus
			\begin{align}
				\rho_{A\tilde{H}}&=\left(\tr_H\circ\mathcal{M}\right)\left(\ketbra{\psi}{\psi}_{I\tilde{I}}\otimes\ketbra{\psi}{\psi}_{K\tilde{K}}\otimes\ketbra{\psi}{\psi}_{J\tilde{J}}\right)\\
				&=\left(\overline{\mathcal{M}}\circ\tr_J\right)\left(\ketbra{\psi}{\psi}_{I\tilde{I}}\otimes\ketbra{\psi}{\psi}_{K\tilde{K}}\otimes\ketbra{\psi}{\psi}_{J\tilde{J}}\right)\\
				&=\overline{\mathcal{M}}\left(\ketbra{\psi}{\psi}_{I\tilde{I}}\otimes\ketbra{\psi}{\psi}_{K\tilde{K}}\right)\otimes\mixed{\tilde{J}}.
			\end{align}
		Using again \cref{eq:psidecomposed2}, we can thus write 
			\begin{align}
				\tr_{B\tilde{B}^z} \circ \Pi_{\tilde{H}^z}(\rho_{AB\tilde{H}})
				&=\tr_{\tilde{B}^z} \circ \Pi_{\tilde{H}^z}(\rho_{A\tilde{H}})\\
				&= \tr_{\tilde{K}_B^z}  \circ \Pi_{\tilde{H}^z} \circ \overline{\mathcal{M}}\left(\ketbra{\psi}{\psi}_{K\tilde{K}} \otimes \ketbra{\psi}{\psi}_{I\tilde{I}} \right) \\
				&=\tr_{\tilde{K}_B^z} \circ \overline{\mathcal{M}}\Bigg({\textstyle \frac{1}{\dim(K)}} \sum_{\substack{a,b,\\ \bar{a},\bar{b}}} \ketbra{a}{\bar{a}}_{K_A^{z}} \otimes \ketbra{b}{\bar{b}}_{K_B^{z}}
			\otimes \ketbra{a}{\bar{a}}_{\tilde{K}_A^{z}} \otimes \ketbra{b}{\bar{b}}_{\tilde{K}_B^{z}}\otimes\ketbra{\psi}{\psi}_{I\tilde{I}}\Bigg)\\
				&={\textstyle \frac{1}{\dim(K)}}\,\overline{\mathcal{M}}\Bigg(\sum_{a,\bar{a}}\ket{a}\bra{\bar{a}}_{K_A^{z}}\otimes\underbrace{\mathrm{id}_{K_B^z}}_{=\dim(K_B^z)\mixed{K_B^z}}\otimes\ketbra{a}{\bar{a}}_{\tilde{K}_A^z}\otimes\ketbra{\psi}{\psi}_{I\tilde{I}}\Bigg)\\
				&={\textstyle \frac{\dim(K_B^z)}{\dim(K)}}\,\xi_{A\tilde{A}^z}^{z}.
			\end{align}
		Comparing this to \cref{eq:tracedrho} yields
			\begin{align}
				\xi_{A\tilde{A}^z}^{z}&={\textstyle \frac{1}{\dim(I)}} \, \tau_{A\tilde{A}^z}^{z}.
			\end{align}
		By an analogous argument we obtain
		        \begin{align}
				\xi_{B\tilde{B}^z}^{z}&={\textstyle \frac{1}{\dim(J)}} \, \tau_{B\tilde{B}^z}^{z}.
			\end{align}
		Inserting this into \cref{eq:CJsigma} and comparing to \cref{eq:rhowithtaus} yields $\xi_{AB\tilde{H}}=\rho_{AB\tilde{H}}$. Thus, the respective C.-J.~states of $\mathcal{N}\circ\mathcal{M}$ and $(\overline{\mathcal{M}}\otimes\overline{\mathcal{N}})\circ\mathcal{D}$ are equal, hence the two maps are equal.
		
		To conclude the proof of \cref{claim:tensormap}, we need to verify that the map $\mathcal{D}$ is TP as claimed. Note first that, by the definition of $\overline{\mathcal{M}}$ and the TP property of~$\mathcal{M}$ (see \cref{rem:TP}),
		\begin{align}
		  \mathrm{tr}_A \circ \overline{\mathcal{M}} \circ \mathrm{tr}_J 
		  = \mathrm{tr}_{A H} \circ \mathcal{M}
		  = \mathrm{tr}_H 
		  = \mathrm{tr}_{K I J}.
		\end{align}
		This implies that $\mathrm{tr}_A \circ \overline{\mathcal{M}}  = \mathrm{tr}_{K I}$, which means that $\overline{\mathcal{M}}$ is TP. Similarly, one can see that $\overline{\mathcal{N}}$ is TP. Hence, $\overline{\mathcal{M}} \otimes \overline{\mathcal{N}}$, which goes from $K' \otimes K'' \otimes I \otimes J$ to $A \otimes B$, is TP, too. Using this, then \cref{eq:mainresult}, and finally that $\mathcal{N} \circ \mathcal{M}$ is TP, we find
		\begin{align}
		  \mathrm{tr}_{K' K'' I J} \circ \mathcal{D} = \mathrm{tr}_{A B} \circ (\overline{\mathcal{M}} \otimes \overline{\mathcal{N}}) \circ \mathcal{D} = \mathrm{tr}_{A B} \circ \mathcal{N} \circ \mathcal{M} 
		  = \mathrm{tr}_{K I J}.
		\end{align}
		This implies that $\mathrm{tr}_{K' K''} \circ \mathcal{D} = \mathrm{tr}_{K}$, i.e., $\mathcal{D}$ is TP.
	\end{proof}

	With the proof of \cref{claim:tensormap}, we have established \cref{eq:mainresult}. 
	\end{proof}

\begin{rem}\label{rem:prooftracenoninc}
        \cref{thm:tensorP} follows from a more specialised version of the same where the maps $\mathcal{M}$ and $\mathcal{N}$ are assumed to be TP. To see this, note first that \cref{lem:unitaltrace} immediately implies that $\tr_A \circ \mathcal{M}$ is TP. Hence, $\mathcal{M}$ must be TP anyway. It thus remains to show the following: For any map $\mathcal{N}$ from $H$ to $B$, there exists a TP map $\mathcal{N}'$ such that the correctness of \cref{thm:tensorP} for $\mathcal{N}'$ implies the correctness of the theorem for $\mathcal{N}$.
                
Let thus $\mathcal{N}$ be a CP map that satisfies the assumptions of~\cref{thm:tensorP}. From~\cref{rem:cpmrescaling} and the fact that $H$ is finite-dimensional, we know that $\mathcal{N}$ can always be rescaled such that it is trace non-increasing. The rescaling does not alter Condition~\ref{cond:1}, Condition~\ref{cond:3}, and~\eqref{eq:mainresult}. We can thus assume without loss of generality that $\mathcal{N}$ is trace non-increasing. Let now $\mathcal{N}'$ be the TP extension of $\mathcal{N}$ defined by \cref{lem:cpmtracepreserving}, which maps from $H$ to $B'$, where $B' = B \oplus \mathrm{span} \{\ket{\perp}\}$. We have
		\begin{align}
			\mathrm{tr}_{A}\circ \mathcal{N}' \circ \mathcal{M}
			 & =\mathrm{tr}_{A}\circ\,\mathcal{N}\circ\mathcal{M}+ \perp_{B'} \circ \big( \mathrm{tr}_A \circ \mathrm{tr}_H \circ \mathcal{M}-\, \mathrm{tr}_A \circ \mathrm{tr}_B \circ \mathcal{N} \circ \mathcal{M} \big) \\ \mathcal{M}-\,\mathrm{tr}_B \circ \mathcal{N} \big) \\
			& =\mathcal{N}+ \perp_{B'} \circ \big(\mathrm{tr}_H -\mathrm{tr}_B \circ \mathcal{N} \big) 
			=\mathcal{N}'.
		\end{align}
	where the second equality holds because $\mathcal{N}$ satisfies Condition~\ref{cond:1}, and because $\mathcal{M}$ is TP. This shows that $\mathcal{N}'$ also satisfies Condition~\ref{cond:1}. Furthermore, because $\mathcal{N}$ satisfies Condition~\ref{cond:3} by assumption, it is independent of~$I$, and hence  \cref{lem:cpmtracepreserving} implies the same is true for  $\mathcal{N}'$. We have thus established that $\mathcal{N}'$ is a CPTP map that meets all conditions of \cref{thm:tensorP}. The specialised version of \cref{thm:tensorP} for TP maps now implies that there exists a CPTP map $\mathcal{D}$ such that
		\begin{equation}
		\mathcal{N}' \circ \mathcal{M} = \left(\overline{\mathcal{M}} \otimes \overline{\mathcal{N}'}\right)\circ\mathcal{D}.
		\end{equation}
	We can concatenate both sides with a projection map $\Pi_B$ onto the subspace $B$ of $B'$. Since $\mathcal{N} = \Pi_B\circ \mathcal{N}'$ and $\overline{\mathcal{N}} = \Pi_B\circ\overline{\mathcal{N}'}$ we find that~\eqref{eq:mainresult} and, hence, \cref{thm:tensorP}, holds for the general map $\mathcal{N}$.  
\end{rem}

\begin{rem} \label{rem:unitalrequirement}
  The unitality requirement in Condition~\ref{cond:2} can be replaced by the weaker condition that $\mathrm{tr}_A \circ \mathcal{M}$ is unital when restricted to the subsystem~$K \otimes J$ of $H$, i.e., 
 \begin{equation} \label{eq:weakunitality}
    \mathrm{tr}_{A I} \circ \mathcal{M}(\rho_{I} \otimes \id_{K J}) = \id_{K J}.
  \end{equation}
  To see this, let $\mathcal{M}$ be such that it satisfies Conditions~\ref{cond:1} and~\ref{cond:3},  as well as \cref{eq:weakunitality}. Furthermore, define a modified map $\mathcal{M'} \coloneqq \mixed{I} \circ \tr_I \circ \mathcal{M}$. Because $\mathcal{N}$ is independent of $I$, $\mathcal{M'}$ also satisfies Condition~\ref{cond:1}. Clearly, it also satisfies Condition~\ref{cond:3}. And because of \cref{eq:weakunitality}, $\mathcal{M'}$ also fulfils Condition~\ref{cond:2}. We can thus apply the theorem to the modified map~$\mathcal{M'}$, which implies that~\cref{eq:mainresult} holds for $\mathcal{M}'$ and $\overline{\mathcal{M'}}$. But $\overline{\mathcal{M'}} = \overline{\mathcal{M}}$, which can be verified by using \cref{rem:mapunique}:
  	\begin{align}
  		\overline{\mathcal{M}'}&=\tr_H\circ\mathcal{M}'\circ\zeta_J\\
  		&=\tr_H\circ\mixed{I}\circ\tr_I\circ\mathcal{M}\circ\zeta_J\\
  		&=\tr_H\circ\mathcal{M}\circ\zeta_J\\
  		&=\overline{\mathcal{M}},
  	\end{align}
  where $\zeta_J$ is the map that creates a state $\zeta_J$ on $J$. Furthermore,  again because $\mathcal{N}$ is independent of $I$ and can thus be written as $\overline{\mathcal{N}}\circ\tr_I$, 
  	\begin{align}
  		\mathcal{N}\circ\mathcal{M}'
  		&=\overline{\mathcal{N}}\circ\tr_I\circ\mixed{I}\circ\tr_I\circ\mathcal{M}\\
  		&=\overline{\mathcal{N}}\circ\tr_I\circ\mathcal{M}\\
  		&=\mathcal{N}\circ\mathcal{M}.
  	\end{align}
  Hence, \cref{eq:mainresult} also holds for~$\mathcal{M}$.
\end{rem}

\begin{rem} \label{rem:infiniteAB}
  It is sufficient to prove \cref{thm:tensorP} for finite-dimensional Hilbert spaces $A$ and $B$, as this implies the general case where these systems have unbounded dimensions. 

  To see this, let $\mathcal{M}$ and $\mathcal{N}$ be CP maps for infinite-dimensional $A$ and $B$ that satisfy Conditions~\ref{cond:1}, \ref{cond:2}, and~\ref{cond:3} of \cref{thm:tensorP}. Furthermore, let $(\Pi_A^d)_{d \in \mathbb{N}}$ be a sequence of CP maps that project on $d$-dimensional nested subspaces of system~$A$, i.e., $\Pi_A^d \circ \Pi_A^{d'} = \Pi_A^{d}$ for any $d \leq d'$ such that, for all states $\rho_A$, 
  \begin{align} \label{eq:projectorsconvergence}
    \lim_{d \to \infty} \Pi^d_A(\rho_A) = \rho_A
  \end{align}
  Similarly, we denote by $(\Pi_B^d)_{d \in \mathbb{N}}$ a sequence of projection maps for the system~$B$. We can then define sequences of CP maps $(\mathcal{M}^d)_{d \in \mathbb{N}}$ and $(\mathcal{N}^d)_{d \in \mathbb{N}}$ by 
 \begin{align}
   \mathcal{M}^d & \coloneqq \bigl( \Pi_A^d + \zeta_A \circ \mathrm{tr}_A \circ (\mathcal{I}_A - \Pi_A^d)\bigr) \circ \mathcal{M} \\
   \mathcal{N}^d & \coloneqq \Pi_B^d \circ \mathcal{N},
  \end{align} 
  where $\zeta_A$ is an arbitrary state on $A$. 
  
  When considering the convergence of sequences of CP maps, we use the topology of their C.-J.~representation as states, which in turn is induced by the trace distance.\footnote{The C.-J.~operators are well-defined because we will only consider maps that take a finite-dimensional input.} Thus, using~\eqref{eq:projectorsconvergence},
  \begin{align} 
    \lim_{d \to \infty} \mathcal{M}^d & = \mathcal{M} \label{eq:Mconvergence} \\
    \lim_{d \to \infty} \mathcal{N}^d & = \mathcal{N} \label{eq:Nconvergence}.
  \end{align}
  Furthermore, for any $d \in \mathbb{N}$ we have
  \begin{align}
    \mathrm{tr}_A \circ \mathcal{M}^d = \mathrm{tr}_A \circ \Pi_A^d \circ \mathcal{M} + \underbrace{\mathrm{tr}_A(\zeta_A)}_{=1} (\mathrm{tr}_A \circ \mathcal{M} - \mathrm{tr}_A \circ \Pi_A^d \circ \mathcal{M})
    = \mathrm{tr}_A \circ \mathcal{M}.
  \end{align}  
  Using this one can readily verify that each of the pairs of maps $\mathcal{M}^d$ and $\mathcal{N}^d$ satisfies Conditions~\ref{cond:1}, \ref{cond:2}, and~\ref{cond:3}. \cref{thm:tensorP} thus implies that there exists a CPTP map $\mathcal{D}^d$ such that 
  \begin{equation} \label{eq:sequenceproduct}
    \mathcal{N}^d \circ \mathcal{M}^d = \bigl( \overline{\mathcal{M}^d} \otimes \overline{\mathcal{N}^d} \bigr) \circ \mathcal{D}^d.
  \end{equation}
  Note that $\mathcal{D}^d$ is a CPTP map from $K$ to $K \otimes K$, which is finite-dimensional. By the C.-J.~isomorphism, the set of such maps is isomorphic to a closed subset of (normalised) states on a finite-dimensional space, which one may also purify. Furthermore, the set of pure states can be continuously embedded into a (real) Euclidean space, where it corresponds to a sphere of radius~$1$. Hence, the set of possible maps $\mathcal{D}^d$ is bounded and closed. We can thus employ the Bolzano-Weierstrass theorem, which tells us that there exists a subsequence of $(\mathcal{D}^d)_{d \in \mathbb{N}}$ that converges to a CPTP map $\mathcal{D}$, i.e., there exists $(d_i)_{i \in \mathbb{N}}$ such that
\begin{equation}
  \mathcal{D} = \lim_{i \to \infty} \mathcal{D}^{d_i} \ .
\end{equation} 

Note that the convergence of the sequences \cref{eq:Mconvergence} and \cref{eq:Nconvergence} also implies the convergence of the subsequences, i.e., $\lim_{i \to \infty} \mathcal{M}^{d_i} = \mathcal{M}$ and $\lim_{i\to \infty} \mathcal{N}^{d_i} = \mathcal{N}$. Using this and \cref{eq:sequenceproduct} we find 
\begin{align}
 \mathcal{N} \circ \mathcal{M} 
 & =  \lim_{i \to \infty} \mathcal{N}^{d_i} \circ \mathcal{M}^{d_i} \\
 & = \lim_{i \to \infty} \ \bigl( \overline{\mathcal{M}^{d_i}} \otimes \overline{\mathcal{N}^{d_i}} \bigr) \circ \mathcal{D}^{d_i} \\
 & = \bigl( \lim_{i \to \infty}  \overline{\mathcal{M}^{d_i}} \otimes \overline{\mathcal{N}^{d_i}} \bigr) \circ \bigl(\lim_{i \to \infty}  \mathcal{D}^{d_i} \bigr) \\
  & = \bigl( \lim_{i \to \infty}  \overline{\mathcal{M}^{d_i}} \otimes \overline{\mathcal{N}^{d_i}} \bigr) \circ \mathcal{D}.
\end{align}
Finally, \cref{rem:mapunique} implies that, for arbitrary states $\zeta_I$ and $\zeta_J$, $\overline{\mathcal{M}^{d_i}} 
  = \mathrm{tr}_H \circ \mathcal{M}^{d_i} \circ \zeta_J$ and $\overline{\mathcal{N}^{d_i}}  = \mathcal{N}^{d_i} \circ \zeta_I$. 
  Hence, 
\begin{align}
\lim_{i \to \infty}  \overline{\mathcal{M}^{d_i}} \otimes \overline{\mathcal{N}^{d_i}} 
& =  \lim_{i \to \infty}  \mathrm{tr}_{H} \circ \mathcal{M}^{d_i} \circ \zeta_J \otimes \mathcal{N}^{d_i} \circ \zeta_I \\
   &=  \lim_{i \to \infty} \bigl( \Pi_A^{d_i} + \zeta_A \circ \mathrm{tr}_A \circ (\mathcal{I}_A - \Pi_A^{d_i})\bigr) \circ \mathrm{tr}_H \circ \mathcal{M} \circ \zeta_J \otimes  \Pi_B^{d_i} \circ \mathcal{N} \circ \zeta_I \\
   & =  \lim_{i \to \infty} \Bigl( \Pi_A^{d_i}  \otimes \Pi_B^{d_i} + \zeta_A \circ \mathrm{tr}_A \circ (\mathcal{I}_A - \Pi_A^{d_i}) \otimes \Pi_B^{d_i} \Bigr) \circ \bigl( \mathrm{tr}_H \circ \mathcal{M} \circ \zeta_J \otimes \mathcal{N} \circ \zeta_I \bigr) \\
   & =   \underbrace{ \lim_{i \to \infty}  \bigl(\Pi_A^{d_i}  \otimes \Pi_B^{d_i} \bigr) \circ  \bigl( \mathrm{tr}_H \circ \mathcal{M} \circ \zeta_J \otimes \mathcal{N} \circ \zeta_I \bigr)}_{=\mathrm{tr}_H \circ \mathcal{M} \circ \zeta_J \otimes \mathcal{N} \circ \zeta_I} \\
   & \qquad + \zeta_A \circ \mathrm{tr}_A \circ \underbrace{ \lim_{i \to \infty} \bigl( (\mathcal{I}_A - \Pi_A^{d_i}) \otimes \Pi_B^{d_i}  \bigr)  \circ \bigl(\mathrm{tr}_H \circ \mathcal{M} \circ \zeta_J \otimes \mathcal{N} \circ \zeta_I \bigr)}_{= 0}\\
   & = \mathrm{tr}_{H} \circ \mathcal{M} \circ \zeta_J \otimes \mathcal{N} \circ \zeta_I \\
   & = \overline{\mathcal{M}} \otimes \overline{\mathcal{N}}.
   \end{align}
 Combining this with the equality above, we obtain~\cref{eq:mainresult}.
\end{rem}

As a preparation for our converse statement, \cref{thm:converse}, we first establish some general properties of the doubling map $\mathcal{D}$ that comes out of  \cref{thm:tensorP}. 

\begin{rem}
	\label{rem:Dprops}
	The CPTP map $\mathcal{D}:K\to K'\otimes K''$ in \cref{thm:tensorP} fulfils
		\begin{enumerate}[label=(\roman*)]
			\item \label{item:lemmacond1}$\tr_{\tilde{K}'}\circ\mathcal{D}_{K''\to \tilde{K}'\tilde{K}''}\circ\mathcal{D}_{K''\to K' K''}=\mathcal{D}_{K\to K' \tilde{K}''}$ and
			\item \label{item:lemmacond2}$\tr_{K'}\circ\mathcal{D}_{K\to K' K''}$ is unital.
		\end{enumerate}
	(Note that all $K$'s are isomorphic spaces, but we use the notation above to distinguish them to keep track of where the different maps go.)

	To prove Property~\ref{item:lemmacond1}, we show how the maps on each side of the equality act on a general basis element $(\ket{a}_{K_A^z}\otimes \ket{b}_{K_B^z})(\bra{\bar{a}}_{K_A^{\bar{z}}} \otimes \bra{\bar{b}}_{K_B^{\bar{z}}})$ of the space of operators on $K$. For the right-hand side, applying $\mathcal{D}_{K\to K' \tilde{K}''}$ as defined in \cref{eq:defD,eq:defV} yields
		\begin{align}
			&\mathcal{D}_{K\to K' \tilde{K}''}\big((\ket{a}_{K_A^z} \otimes \ket{b}_{K_B^z})(\bra{\bar{a}}_{K_A^{\bar{z}}} \otimes \bra{\bar{b}}_{K_B^{\bar{z}}})\big)\\
			&\hspace{20pt}=\sum_{\tilde{z},\tilde{a},\tilde{b},\tilde{\tilde{a}},\tilde{\tilde{b}}}\big((\ket{\tilde{a}}_{K_A'^{\tilde{z}}} \otimes \ket{\tilde{b}}_{\tilde{K}_B''^{\tilde{z}}})({\color{myred}\bra{\tilde{a}}_{K_A^{\tilde{z}}}} \otimes {\color{myblue}\bra{\tilde{b}}_{K_B^{\tilde{z}}}})\big)\big(({\color{myred}\ket{a}_{K_A^{z}}} \otimes {\color{myblue}\ket{b}_{K_B^z}})({\color{myorange}\bra{\bar{a}}_{K_A^{\bar{z}}}} \otimes {\color{mygreen}\bra{\bar{b}}_{K_B^{\bar{z}}}})\big)\\
			&\hspace{80pt}\big(({\color{myorange}\ket{\tilde{\tilde{a}}}_{K_A^{\tilde{z}}}} \otimes {\color{mygreen}\ket{\tilde{\tilde{b}}}_{K_B^{\tilde{z}}}})(\bra{\tilde{\tilde{a}}}_{K_A'^{\tilde{z}}} \otimes \bra{\tilde{\tilde{b}}}_{\tilde{K}_B''^{\tilde{z}}})\big)\otimes\mixed{K_B'^{\tilde{z}}}\otimes\mixed{\tilde{K}_A''^{\tilde{z}}}\\
			&\hspace{20pt}=\label{eq:lemmarhs}\ketbra{a}{\bar{a}}_{K_A'^z}\otimes\mixed{K_B'^z}\otimes\mixed{\tilde{K}_A''^z}\otimes\ketbra{b}{\bar{b}}_{\tilde{K}_B''^z}\cdot\delta_{z\bar{z}},
		\end{align}
	where the colours indicate which parts combined yield a $\delta$-function on the respective labels. The left-hand side of~\ref{item:lemmacond1} applied to the same bases element yields 
		\begin{align}
			&\tr_{\tilde{K}'}\circ\mathcal{D}_{K''\to \tilde{K}'\tilde{K}''}\circ\mathcal{D}_{K''\to K' K''}\big((\ket{a}_{K_A^z}\otimes \ket{b}_{K_B^z})(\bra{\bar{a}}_{K_A^{\bar{z}}} \otimes \bra{\bar{b}}_{K_B^{\bar{z}}})\big)\\
			&\hspace{20pt}=\tr_{\tilde{K}'}\circ\mathcal{D}_{K''\to \tilde{K}'\tilde{K}''}\Big(\underbrace{\ketbra{a}{\bar{a}}_{K_A'^z}\otimes\mixed{K_B'^z}}_{\eqqcolon\rho_{K'}}\otimes\underbrace{\mixed{K_A''^z}}_{=1/\dim(K_A^z)\sum_{\hat{a}}\ketbra{\hat{a}}{\hat{a}}_{K_A''^z}}\otimes\ketbra{b}{\bar{b}}_{K_B''^z}\Big)\cdot\delta_{z\bar{z}}\\
			&\hspace{20pt}=\tr_{\tilde{K}'}\Big({\textstyle \frac{1}{\dim(K_A^z)}} \sum_{\hat{a}}\mathcal{D}_{K''\to \tilde{K}'\tilde{K}''}\big(\ketbra{\hat{a}}{\hat{a}}_{K_A''^z}\otimes\ketbra{b}{\bar{b}}_{K_B''^z}\big)\Big)\otimes\rho_{K'}\cdot\delta_{z\bar{z}}\\
			&\hspace{20pt}=\tr_{\tilde{K}'}\Big(\underbrace{{\textstyle \frac{1}{\dim(K_A^z)}}\sum_{\hat{a}}\ketbra{\hat{a}}{\hat{a}}_{\tilde{K}_A'^z}}_{=\mixed{\tilde{K}_A'^z}}\otimes\ketbra{b}{\bar{b}}_{\tilde{K}_B''^z}\otimes\mixed{\tilde{K}_B'^z}\otimes\mixed{\tilde{K}_A''^z}\Big)\otimes\rho_{K'}\cdot\delta_{z\bar{z}}\\
			&\hspace{20pt}=\ketbra{a}{\bar{a}}_{K_A'^z}\otimes\mixed{K_B'^z}\otimes\mixed{\tilde{K}_A''^z}\otimes\ketbra{b}{\bar{b}}_{\tilde{K}_B''^z}\cdot\delta_{z\bar{z}},
		\end{align}
	which is equal to \cref{eq:lemmarhs}. Since this is true for any basis element of the space of operators on $K$, the maps are equal, proving~\ref{item:lemmacond1}. Property~\ref{item:lemmacond2} can be shown by a direct calculation, using that $\mathrm{id}_K=\sum_{z,a,b}\ketbra{a}{a}_{K_A^z}\otimes\ketbra{b}{b}_{K_B^z}$:
			\begin{align}
			\tr_{K'}\circ\mathcal{D}_{K\to K' K''}(\mathrm{id}_K)&=\tr_{K'}\Big(\sum_{a,b,z}\mathcal{D}_{K\to K' K''}\big(\ketbra{a}{a}_{K_A^z}\otimes\ketbra{b}{b}_{K_B^z}\big)\Big)\\
			&=\tr_{K'}\Big(\sum_{a,b,z}\ketbra{a}{a}_{K_A'^z}\otimes\mixed{K_B'^z}\otimes\mixed{K_A''^z}\otimes\ketbra{b}{b}_{K_B''^z}\Big)\\
			&=\sum_z\underbrace{\tr_{K'}(\mathrm{id}_{K_A'^z}\otimes\mixed{K_B'^z}}_{=\dim(K_A'^z)})\mixed{K_A''^z}\otimes\mathrm{id}_{K_B''^z}\\
			&=\sum_z\mathrm{id}_{K_A''^z}\otimes\mathrm{id}_{K_B''^z}\\
			&=\mathrm{id}_{K''}.
		\end{align}
\end{rem}

We are now ready to state and prove the converse to \cref{thm:tensorP}. Identifying $\mathcal{A}$ and $\mathcal{B}$ with $\overline{\mathcal{M}}$ and $\overline{\mathcal{N}}$, respectively, it implies that the Conditions~\ref{cond:1}, \ref{cond:2}, and~\ref{cond:3} are necessary. 

\begin{thm}[Converse statement to \cref{thm:tensorP}] \label{thm:converse}
	Let $\mathcal{E}:I\otimes K\otimes J\to A\otimes B$ be a CP map such that $\mathcal{E}=(\mathcal{A}\otimes\mathcal{B})\circ\mathcal{D}$ for a CPTP map $\mathcal{A}:I\otimes K\to A$ and CP maps $\mathcal{B}:K\otimes J\to B$ and $\mathcal{D}:K\to K\otimes K$, where the latter fulfils Properties~\ref{item:lemmacond1} and~\ref{item:lemmacond2} in \cref{rem:Dprops}.\footnote{The requirement that $\mathcal{A}$ is TP does not limit the validity of this theorem as a converse statement to \cref{thm:tensorP}. Indeed, according to \cref{rem:prooftracenoninc}, the map $\mathcal{M}$ that enters \cref{thm:tensorP} is anyway TP, and hence also $\overline{\mathcal{M}}$.} 
	
	Then there exist CP maps $\mathcal{M}:I\otimes K\otimes J\to A\otimes I\otimes K\otimes J$ and $\mathcal{N}:I\otimes K\otimes J\to B$ such that $\mathcal{E}=\mathcal{N}\circ\mathcal{M}$ and $\mathcal{M}$ and $\mathcal{N}$ fulfil Conditions~\ref{cond:1}, \ref{cond:2}, and \ref{cond:3} in \cref{thm:tensorP}. 
\end{thm}

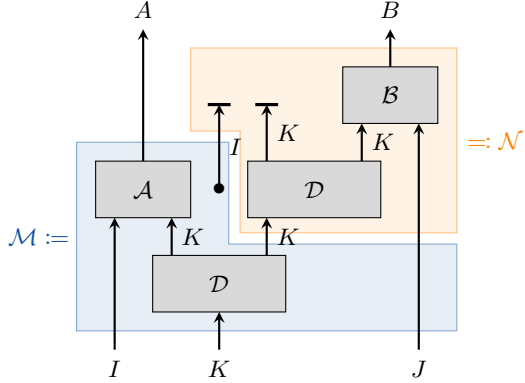
\begin{figure}[t]
	\centering
	\begin{tikzpicture}
		\draw[fill=myblue!10,draw=myblue!50] (-0.25,-1.5) -- (4.75,-1.5) -- (4.75,-0.35) -- (1.75,-0.35) -- (1.75,1) -- (-0.25,1) -- cycle;
		\draw[fill=myorange!10,draw=myorange!50] (4.75,2.25) -- (4.75,-0.2) -- (1.9,-0.2) -- (1.9,1.15) -- (1.25,1.15) -- (1.25,2.25) -- cycle;
		\draw[fill=gray!30] (0,0) rectangle (1.25,0.75);
		\node at (0.625,0.375) {\small $\mathcal{A}$};
		\begin{scope}[xshift=3.25cm,yshift=1.25cm]
		\draw[fill=gray!30] (0,0) rectangle (1.25,0.75);
		\node at (0.625,0.375) {\small $\mathcal{B}$};
		\end{scope}
		\draw[fill=gray!30] (0.75,-1.25) rectangle (2.5,-0.5);
		\node at (1.625,-0.875) {\small $\mathcal{D}$};
		\begin{scope}[xshift=0.375cm]
		\draw[fill=gray!30] (1.625,0) rectangle (3.375,0.75);
		\node at (2.5,0.375) {\small $\mathcal{D}$};
		\end{scope}
		\draw[->,>=stealth, thick] (0.25,-1.75) -- (0.25,0);
		\draw[->,>=stealth, thick] (4.25,-1.75) -- (4.25,1.25);
		\draw[->,>=stealth, thick] (1.625,-1.75) -- (1.625,-1.25);
		\draw[->,>=stealth, thick] (1,-0.5) to node[right] {\small $K$} (1,0);
		\draw[->,>=stealth, thick] (2.25,-0.5) to node[right] {\small $K$} (2.25,0);
		\begin{scope}[xshift=1.25cm,yshift=1.25cm]
		\draw[->,>=stealth, thick] (1,-0.5) to node[right] {\small $K$} (1,0.25);
		\draw[->,>=stealth, thick] (2.25,-0.5) to node[right] {\small $K$} (2.25,0);
		\end{scope}
		\draw[->,>=stealth, thick] (0.625,0.75) -- (0.625,2.5);
		\begin{scope}[xshift=3.25cm,yshift=1.25cm]
		\draw[->,>=stealth, thick] (0.625,0.75) -- (0.625,1.25);
		\end{scope}
		\node at (0.25,-2) {\small $I$};
		\node at (1.625,-2) {\small $K$};
		\node at (4.25,-2) {\small $J$};
		\node at (0.625,2.75) {\small $A$};
		\node at (3.875,2.75) {\small $B$};
		\draw[very thick] (2.1,1.5) -- (2.4,1.5);
		\node at (1.625,0.375) {\textbullet};
		\draw[->,>=stealth, thick] (1.625,0.375) to node[right] {\small $I$} (1.625,1.5);
		\draw[very thick] (1.475,1.5) -- (1.775,1.5);
		\node[color=myblue] at (-0.75,-0.25) {\small $\mathcal{M}\coloneqq$};
		\node[color=myorange] at (5.25,1.025) {\small $\eqqcolon\mathcal{N}$};
	\end{tikzpicture}
	\caption{\label{fig:conversethm}\textbf{Visualisation of the map $\mathbf{\mathcal{E}}$ occurring in \cref{thm:converse}.} The diagram shows the components of the map $\mathcal{E}$ as given in \cref{eq:MNforconverse}. The blue and orange boxes define the maps $\mathcal{M}$ and $\mathcal{N}$, respectively.}
\end{figure}

\begin{proof}
	First, note that we can insert a map of the form $\tr_I\circ\mixed{I}$ into $(\mathcal{A}\otimes\mathcal{B})\circ\mathcal{D}$ without changing the map. Together with Property~\ref{item:lemmacond1} in \cref{rem:Dprops} this allows us to write (see \cref{fig:conversethm})
		\begin{align}
			\mathcal{E}=(\mathcal{A}\otimes\mathcal{B})\circ\mathcal{D}&=(\mathcal{A}\otimes\mathcal{B})\circ\tr_K\circ\mathcal{D}\circ\mathcal{D}\\
			&=\mathcal{A}\circ\mathcal{B}\circ\tr_K\circ\mathcal{D}\circ\tr_I\circ\mixed{I}\circ\mathcal{D}\\
			&=\underbrace{\big(\tr_K\circ\mathcal{B}\circ\mathcal{D}\circ\tr_I\big)}_{\eqqcolon \mathcal{N}}\circ\underbrace{\big(\mixed{I}\circ\mathcal{A}\circ\mathcal{D}\big)}_{\eqqcolon \mathcal{M}}.\label{eq:MNforconverse}
		\end{align}
	Note that $\mathcal{M}$ acts on $J$ as the identity. We can then show that $\mathcal{M}$ and $\mathcal{N}$ fulfil Conditions~\ref{cond:1}, \ref{cond:2}, and \ref{cond:3} in \cref{thm:tensorP}: 
		\begin{enumerate}[label=(\roman*)]
			\item 
					$\tr_A\circ\mathcal{N}\circ\mathcal{M}
					=\tr_A\circ\big(\mathcal{A}\otimes\mathcal{B}\big)\circ\mathcal{D}
					=\tr_K\circ\mathcal{B}\circ\mathcal{D}\circ\tr_I
					=\mathcal{N}$.
			\item Property~\ref{item:lemmacond2} in \cref{rem:Dprops} says that $\tr_K\circ\mathcal{D}$ is unital. Hence,
				\begin{align}
					\tr_A\circ\mathcal{M}(\mathrm{id}_{IKJ})&=\tr_A\circ\mixed{I}\circ\mathcal{A}\circ\mathcal{D}(\mathrm{id}_{IKJ})\\
					&=\mixed{I}\circ\underbrace{\tr_I(\mathrm{id}_I)}_{=\dim(I)}\otimes\underbrace{\tr_K\circ\mathcal{D}(\mathrm{id}_K)}_{\mathrm{id}_K}\otimes\mathrm{id}_J\\
					&=\mathrm{id}_{IKJ},
				\end{align}
				thus $\tr_A\circ\mathcal{M}$ is unital.
			\item From the definition of $\mathcal{M}$ and $\mathcal{N}$ it directly follows that $\mathrm{tr}_H \circ \mathcal{M}$ and $\mathcal{N}$ are independent of $J$ and $I$, respectively.
		\end{enumerate}
\end{proof}

\begin{rem}
  It has been shown in~\cite{Lorenz2021} that, if a map is unitary and satisfies a non-signalling condition analogous to Condition~\ref{cond:3} of \cref{thm:tensorP}, then this map has a structure similar to the right-hand side of \cref{eq:mainresult}. Note that the unitarity assumption is crucial for this result: the PR box~\cite{PopescuRohrlich94} does not admit such a structure, although it satisfies the non-signalling condition. In contrast, \cref{thm:tensorP} is valid for general (not necessarily unitary) CP maps, but instead requires the additional Conditions~\ref{cond:1} and~\ref{cond:2} (which are also necessary; see \cref{thm:converse}). 
\end{rem}

%% file: 3_independence.tex
\section{Implications}
\label{sec:corollaries}

Having established \cref{thm:tensorP}, we can give an answer to the question posed in the introduction, generalising  Tsirelson's result~\cite{Tsirelson2006} to the ``fully quantum'' case. We state this answer as \cref{cor:commute}. \Cref{fig:corollary} illustrates the main assumption of the corollary---a commutation relation between the maps ${\mathcal{X}}$ and ${\mathcal{Y}}$ described as Condition~\ref{cond:c1}---as well as the conclusion, which is that the concatenation of these two maps factorises; see \cref{eq:corollaryresult}.

The special case of Tsirelson's result, which we discuss later as \cref{cor:vNalgebras}, refers to families of measurement operators $\{X_{i,\alpha}\}$ and $\{Y_{j,\beta}\}$ instead of CP maps ${\mathcal{X}}$ and ${\mathcal{Y}}$.  Hence, Condition~\ref{cond:c1} of \cref{cor:commute} can be understood as a quantum generalisation of Tsirelson's assumption that the families of measurement operators $\{X_{i,\alpha}\}$ and $\{Y_{j,\beta}\}$ commute; see \cref{eq:commutingoperators}. Note that the measurement operators satisfy the property $\smash{\sum_\alpha} X_{i,\alpha}=\mathrm{id}_{K}$ and $\smash{\sum_\beta} Y_{j,\beta}=\mathrm{id}_{K}$. In \cref{cor:commute}, this property generalises to a unitality assumption, phrased as Condition~\ref{cond:c2}.\footnote{In fact, Condition~\ref{cond:c2} corresponds to a slightly weaker assumption, for only one of the two maps needs to satisfy the unitality requirement.} This assumption is necessary; if we drop it without replacement, the statement is false, even when one restricts it to the purely classical case. This can be seen by choosing the maps~${\mathcal{X}}$ and ${\mathcal{Y}}$ such that $\mathrm{tr}_H\circ{\mathcal{Y}}\circ{\mathcal{X}}$ implements the PR~box~\cite{PopescuRohrlich94}. It is known that the PR~box does not factorise, for this would violate the Tsirelson bound~\cite{Tsirelson1980}.\footnote{Tsirelson's bound should not be confused with Tsirelson's result~\cite{Tsirelson2006}, which we describe as \cref{cor:vNalgebras}.} We refer to \cref{app:unitality} for more details. 

Furthermore, the statement of \cref{thm:tensorP} can be generalised to a family consisting of more than two maps which fulfil assumptions similar to Conditions~ \ref{cond:1}--\ref{cond:3} in \cref{thm:tensorP}. We present this statement as \cref{cor:multimap}. 

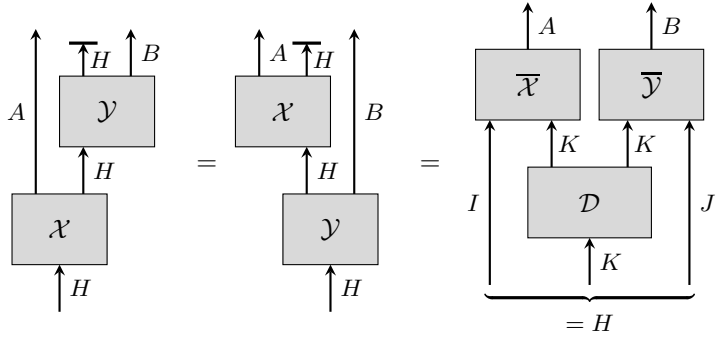
\begin{figure}[t]
	\centering
	\begin{tikzpicture}[baseline=(current bounding box.center),scale=1.25]
		\draw[fill=gray!30] (0,0) rectangle (1,0.75);
		\draw[fill=gray!30] (0.5,1.25) rectangle (1.5,2);
		\node at (0.5,0.375) {$\mathcal{X}$};
		\node at (1,1.625) {$\mathcal{Y}$};
		\draw[thick,->,>=stealth] (0.5,-0.5) to node[right] {\small $H$} (0.5,0);
		\draw[thick,->,>=stealth] (0.25,0.75) to node[left] {\small $A$} (0.25,2.5);
		\draw[thick,->,>=stealth] (0.75,0.75) to node[right] {\small $H$} (0.75,1.25);
		\draw[thick,->,>=stealth] (1.25,2) to node[right] {\small $B$} (1.25,2.5);
		\draw[thick,->,>=stealth] (0.75,2) to node[right=-0.05] {\small $H$} (0.75,2.35);
		\draw[very thick] (0.6,2.35) -- (0.9,2.35);
	\end{tikzpicture}
	\ \ =\ \ 
	\begin{tikzpicture}[xscale=-1,baseline=(current bounding box.center),scale=1.25]
		\draw[fill=gray!30] (0,0) rectangle (1,0.75);
		\draw[fill=gray!30] (0.5,1.25) rectangle (1.5,2);
		\node at (0.5,0.375) {${\mathcal{Y}}$};
		\node at (1,1.625) {${\mathcal{X}}$};
		\draw[thick,->,>=stealth] (0.5,-0.5) to node[right,pos=0.5] {\small $H$} (0.5,0);
		\draw[thick,->,>=stealth] (0.25,0.75) to node[right,pos=0.5] {\small $B$} (0.25,2.5);
		\draw[thick,->,>=stealth] (0.75,0.75) to node[right,pos=0.5] {\small $H$} (0.75,1.25);
		\draw[thick,->,>=stealth] (1.25,2) to node[right] {\small $A$} (1.25,2.5);
		\draw[thick,->,>=stealth] (0.75,2) to node[right=-0.05] {\small $H$} (0.75,2.35);
		\draw[very thick] (0.6,2.35) -- (0.9,2.35);
	\end{tikzpicture}
	\ \ =\ \ 
	\begin{tikzpicture}[baseline=(current bounding box.center),scale=1.25]
		\draw[fill=gray!30] (0.1,0) rectangle (1.4,0.75);
		\node at (0.75,0.375) {$\mathcal{D}$};
		\draw[thick,->,>=stealth] (0.75,-0.5) to node[right,pos=0.5] {\small $K$} (0.75,0);
		\draw[thick,->,>=stealth] (-0.3,-0.5) to node[left,pos=0.5] {\small $I$} (-0.3,1.25);
		\draw[thick,->,>=stealth] (1.8,-0.5) to node[right,pos=0.5] {\small $J$} (1.8,1.25);
		\draw[thick,->,>=stealth] (0.35,0.75) to node[right=-0.05cm,pos=0.5] {\small $K$} (0.35,1.25);
		\draw[thick,->,>=stealth] (1.15,0.75) to node[right=-0.05cm,pos=0.5] {\small $K$} (1.15,1.25);
		\draw[fill=gray!30] (-0.45,1.25) rectangle (0.65,2);
		\draw[fill=gray!30] (0.85,1.25) rectangle (1.95,2);
		\node at (0.1,1.625) {$\overline{\mathcal{X}}$};
		\node at (1.4,1.625) {$\overline{\mathcal{Y}}$};
		\draw[thick,->,>=stealth] (0.1,2) to node[right,pos=0.5] {\small $A$} (0.1,2.5);
		\draw[thick,->,>=stealth] (1.4,2) to node[right,pos=0.5] {\small $B$} (1.4,2.5);
		\draw[very thick, decorate,decoration={calligraphic brace,mirror}] (-0.35,-0.6) -- (1.85,-0.6);
		\node at (0.75,-0.9) {\small $=H$};
	\end{tikzpicture}
	\caption{\label{fig:corollary}\textbf{Visualisation of \cref{cor:commute}.} Condition~\ref{cond:c1} holds if and only if $\mathcal{X}$ and $\mathcal{Y}$ commute, in the sense that the two circuit diagrams on the left have the same input-output behaviour. Provided that the other conditions are also satisfied, the corollary implies that the circuit diagram shown to the right, where $\overline{\mathcal{X}}$ and $\overline{\mathcal{Y}}$ act on two separate copies of $K$, also has the same input-output behaviour. The diagram thus captures the idea that commuting maps factorise.}
\end{figure}

\begin{cor}
	\label{cor:commute}
	Let ${\mathcal{X}}:H\to H\otimes A$ and ${\mathcal{Y}}:H\to H\otimes B$ be CPTP maps, where $H= I\otimes K\otimes J$ is finite-dimensional, such that
	\begin{enumerate}[label=(\roman*)]
		\item 
		$\mathrm{tr}_H\circ{\mathcal{Y}}\circ{\mathcal{X}}=\mathrm{tr}_H\circ{\mathcal{X}}\circ{\mathcal{Y}}$ \label{cond:c1}
		\item either $\mathrm{tr}_A\circ{\mathcal{X}}$ or $\mathrm{tr}_B\circ{\mathcal{Y}}$ is unital \label{cond:c2}
		\item $\mathrm{tr}_H\circ{\mathcal{X}}$ is independent of $J$ and $\mathrm{tr}_H\circ{\mathcal{Y}}$ is independent of $I$. \label{cond:c3}
	\end{enumerate}
	Then there exists a CPTP map $\mathcal{D}:K\to K\otimes K$ such that
	\begin{equation} \label{eq:corollaryresult}
		\mathrm{tr}_H\circ{\mathcal{Y}}\circ{\mathcal{X}}=\big(\overline{\mathcal{X}}\otimes\overline{\mathcal{Y}}\big)\circ\mathcal{D},
	\end{equation}
	where $\overline{\mathcal{X}}\circ\tr_J=\tr_H\circ\mathcal{X}$, $\overline{\mathcal{Y}}\circ\tr_I=\tr_H\circ\mathcal{Y}$.
\end{cor}

\begin{proof}
	Without loss of generalisation, we assume that $\mathrm{tr}_A\circ{\mathcal{X}}$ is unital. If instead $\mathrm{tr}_B\circ{\mathcal{Y}}$ is unital the same proof works by exchanging the roles of ${\mathcal{X}}$ and ${\mathcal{Y}}$. To apply \cref{thm:tensorP}, we set $\mathcal{M}\coloneqq{\mathcal{X}}$ and $\mathcal{N}\coloneqq\mathrm{tr}_H\circ{\mathcal{Y}}$, thus
	$\mathcal{N}\circ\mathcal{M}=\mathrm{tr}_H\circ{\mathcal{Y}}\circ{\mathcal{X}}$. Because ${\mathcal{X}}$ and ${\mathcal{Y}}$ commute under the trace over $H$, it follows that 
		\begin{equation}
			\mathrm{tr}_A\circ\mathcal{N}\circ\mathcal{M}=\mathrm{tr}_{AH}\circ{\mathcal{Y}}\circ{\mathcal{X}}=\mathrm{tr}_{AH}\circ{\mathcal{X}}\circ{\mathcal{Y}}
			=\mathrm{tr}_H\circ{\mathcal{Y}}
			=\mathcal{N},
		\end{equation}
	hence Condition~\ref{cond:1} in \cref{thm:tensorP} is fulfilled. Conditions~\ref{cond:2} and \ref{cond:3} are directly fulfilled by the definition of the maps ${\mathcal{X}},{\mathcal{Y}}$. Hence, we can apply \cref{thm:tensorP} and the statement of the corollary directly follows.
\end{proof}

If we specialise \cref{cor:commute} to the case where the inputs $I$, $J$ and the outputs $A$, $B$ are classical, we retrieve the statement of~\cite{Tsirelson2006} as described above. In fact, one may more generally consider a classical version of \cref{thm:tensorP} instead of \cref{cor:commute}. This yields another generalisation of Tsirelson's result, stated in \cref{rem:TsirelsonGeneralised}, which may be of independent interest. 

\begin{cor}
	\label{cor:vNalgebras}
	Let $\{X_{i,\alpha}\}$ and $\{Y_{j,\beta}\}$ be finite families of positive operators on a finite-dimensional Hilbert space $K$ such that 
		\begin{equation} \label{eq:commutingoperators}
			[X_{i,\alpha},Y_{j,\beta}]=0  \quad \forall \, i,j,\alpha,\beta,
		\end{equation}
	and $\sum_\alpha X_{i,\alpha}=\mathrm{id}_{K}$ and $\sum_\beta Y_{j,\beta}=\mathrm{id}_{K}$ for all $i,j$. Then there exists another finite-dimensional Hilbert space $\overline{K}$ with decomposition  $\overline{K}=K_A\otimes K_B$ and an isometry $V: K\to\overline{K}$ such that
		\begin{equation} \label{eq:operatorfactorisation}
			X_{i,\alpha}= V^* \bigl(A_{i,\alpha}\otimes\mathrm{id}_{K_B}\bigr) V ,\hspace{10pt} Y_{j,\beta}= V^* \bigl(\mathrm{id}_{K_A}\otimes B_{j,\beta} \bigr) V,
		\end{equation} 
	where $A_{i,\alpha}$ and $B_{j,\beta}$ are operators on $K_A$ and $K_B$, respectively.
\end{cor}

\begin{proof}
	The first step in the proof is to apply \cref{cor:commute} to the setting described above, therefore we have to identify the corresponding Hilbert spaces and maps and show that they fulfil the conditions of the theorem. Let
		\begin{align}
			I&\coloneqq\mathrm{span}\{\ket{i}\}_i\\
			J&\coloneqq\mathrm{span}\{\ket{j}\}_j\\
			H&\coloneqq I\otimes K\otimes J\\
			A&\coloneqq\mathrm{span}\{\ket{\alpha}\}_{\alpha}\\
			B&\coloneqq\mathrm{span}\{\ket{\beta}\}_{\beta},
		\end{align}
	where $\{\ket{i}\}_i,\{\ket{j}\}_j,\{\ket{\alpha}\}_{\alpha}$, and $\{\ket{\beta}\}_{\beta}$ are orthonormal families of vectors.
	Define the maps ${\mathcal{X}}:H\to A\otimes H$ and ${\mathcal{Y}}:H\to B\otimes H$ via
		\begin{align}
			\label{eq:defextAB}
			\begin{split}
			{\mathcal{X}}&: W_{H}\mapsto\sum_{i,\alpha}\left(\ketbra{i}{i}_{I}\otimes\sqrt{X_{i,\alpha}}\otimes\mathrm{id}_{J}\right) W_{H}\left(\ketbra{i}{i}_{I}\otimes\sqrt{X_{i,\alpha}}\otimes\mathrm{id}_{J}\right)\otimes\ketbra{\alpha}{\alpha}_A\\
			{\mathcal{Y}}&: W_{H}\mapsto\sum_{j,\beta}\left(\mathrm{id}_{I}\otimes\sqrt{Y_{j,\beta}}\otimes\ketbra{j}{j}_{J}\right) W_{H}\left(\mathrm{id}_{I}\otimes\sqrt{Y_{j,\beta}}\otimes\ketbra{j}{j}_{J}\right)\otimes\ketbra{\beta}{\beta}_B.
			\end{split}
		\end{align}
	These maps are indeed CPTP maps, which can be shown by identifying their respective Kraus operators. We demonstrate this here for the map ${\mathcal{X}}$: Let ${\mathcal{X}}(W_H)=\sum_{i,\alpha} E_{i,\alpha}W_H E_{i,\alpha}^*$, where $E_{i,\alpha}:H\to A\otimes H$ are the Kraus operators of ${\mathcal{X}}$ given by 
		\begin{equation}
			E_{i,\alpha}\coloneqq\ketbra{i}{i}_{I}\otimes\sqrt{X_{i,\alpha}}\otimes\mathrm{id}_{J}\otimes\ket{\alpha}_A.
		\end{equation}
	The set $\{E_{i,\alpha}\}$ forms indeed a valid set of Kraus operators of a TP map:
		\begin{align} \label{eq:KrausCPTP}
			\sum_{i,\alpha}E_{i,\alpha}^*E_{i,\alpha}&=\sum_{i,\alpha}\big(\ketbra{i}{i}_I\otimes\sqrt{X_{i,\alpha}}\otimes\mathrm{id}_{J}\otimes\bra{\alpha}_A\big)\big(\ketbra{i}{i}_I\otimes\sqrt{X_{i,\alpha}}\otimes\mathrm{id}_{J}\otimes\ket{\alpha}_A\big)\\
			&=\sum_{i,\alpha}\ketbra{i}{i}_I\otimes{X_{i,\alpha}}\otimes\mathrm{id}_{J}\otimes\underbrace{\langle \alpha\ket{\alpha}_A}_{=1}\\
			&=\sum_{i}\ketbra{i}{i}_I\otimes\sum_{\alpha}{X_{i,\alpha}}\otimes\mathrm{id}_{J}\\
			&=\sum_i \ketbra{i}{i}_I\otimes\mathrm{id}_K\otimes\mathrm{id}_{J}\\
			&=\mathrm{id}_{H}.
		\end{align}
	The same calculation can be done for ${\mathcal{Y}}$. Hence, the maps are indeed CPTP maps. 

	The definition of the maps in \cref{eq:defextAB} directly allows us to show that the conditions in \cref{cor:commute} are fulfilled: 
	\begin{enumerate}
		\item The maps commute: From $\left[X_{i,\alpha},Y_{j,\beta}\right]=0$ it follows that $\left[\sqrt{X_{i,\alpha}},\sqrt{Y_{j,\beta}}\right]=0$ for all $i,j,\alpha,\beta$. Hence, 
		\begin{align}
			&\mathrm{tr}_H\circ{\mathcal{Y}}\circ{\mathcal{X}}(W_H)\\
			&=\mathrm{tr}_H\Big(\sum_{i,\alpha,j,\beta}\big(\ketbra{i}{i}_I\otimes\sqrt{Y_{j,\beta}}\sqrt{X_{i,\alpha}}\otimes\ketbra{j}{j}_J\big)W_H\big(\ketbra{i}{i}_I\otimes\sqrt{Y_{j,\beta}}\sqrt{X_{i,\alpha}}\otimes\ketbra{j}{j}_J\big)\\
			&\hspace{70pt}\otimes\ketbra{\alpha}{\alpha}_A\otimes\ketbra{\beta}{\beta}_B\Big)\\
			&=\mathrm{tr}_H\Big(\sum_{i,\alpha,j,\beta}\big(\ketbra{i}{i}_I\otimes\sqrt{X_{i,\alpha}}\sqrt{Y_{j,\beta}}\otimes\ketbra{j}{j}_J\big)W_H\big(\ketbra{i}{i}_I\otimes\sqrt{X_{i,\alpha}}\sqrt{Y_{j,\beta}}\otimes\ketbra{j}{j}_J\big)\\
			&\hspace{70pt}\otimes\ketbra{\alpha}{\alpha}_A\otimes\ketbra{\beta}{\beta}_B\Big)\\
			&=\mathrm{tr}_H\circ{\mathcal{X}}\circ{\mathcal{Y}}(W_H).
		\end{align}
		\item $\mathrm{tr}_A\circ{\mathcal{X}}$ is unital:
		\begin{align} \label{eq:Abarunital}
			&\mathrm{tr}_A\circ{\mathcal{X}}(\mathrm{id}_H)\\
			&=\mathrm{tr}_A\left(\sum_{i,\alpha}\left(\ketbra{i}{i}_I\otimes\sqrt{X_{i,\alpha}}\otimes\mathrm{id}_{J}\right) \mathrm{id}_H\left(\ketbra{i}{i}_I\otimes\sqrt{X_{i,\alpha}}\otimes\mathrm{id}_{J}\right)\otimes\ketbra{\alpha}{\alpha}_A\right)\\
			&=\sum_{i,\alpha}\ketbra{i}{i}_I\otimes X_{i,\alpha}\otimes\mathrm{id}_{J}\\
			&=\mathrm{id}_H.
		\end{align}
		\item $\mathrm{tr}_H\circ{\mathcal{X}}$ is independent of $J$, i.e., there exists a map $\overline{\mathcal{X}}:I\otimes K\to A$ such that $\mathrm{tr}_H\circ{\mathcal{X}}=\overline{\mathcal{X}}\circ\mathrm{tr}_{J}$:
			\begin{align}
				&\mathrm{tr}_H\circ{\mathcal{X}}(W_H)\\
				&=\mathrm{tr}_{H}\sum_{i,\alpha}\left(\ketbra{i}{i}_I\otimes\sqrt{X_{i,\alpha}}\otimes\mathrm{id}_{J}\right) W_H\left(\ketbra{i}{i}_I\otimes\sqrt{X_{i,\alpha}}\otimes\mathrm{id}_{J}\right)\otimes\ketbra{\alpha}{\alpha}_A\\
				&=\mathrm{tr}_{IK}\sum_{i,\alpha}\left(\ketbra{i}{i}_I\otimes\sqrt{X_{i,\alpha}}\right) \mathrm{tr}_{J}(W_H)\left(\ketbra{i}{i}_I\otimes\sqrt{X_{i,\alpha}}\right)\otimes\ketbra{\alpha}{\alpha}_A\\
				&=\overline{\mathcal{X}}\circ\mathrm{tr}_{J}(W_H)
			\end{align}
		with $\overline{\mathcal{X}}(W_{IK})\coloneqq\mathrm{tr}_{IK}\Big(\sum_{i,\alpha}\left(\ketbra{i}{i}_I\otimes\sqrt{X_{i,\alpha}}\right) W_{IK}\left(\ketbra{i}{i}_I\otimes\sqrt{X_{i,\alpha}}\right)\otimes\ketbra{\alpha}{\alpha}_A\Big)$. The statement that $\mathrm{tr}_H\circ{\mathcal{Y}}$ is independent of $I$ can be shown analogously.
	\end{enumerate}
	
	Hence, all conditions in \cref{cor:commute} are fulfilled and it follows that there exists a CPTP map $\mathcal{D}:K\to K\otimes K$ such that
		\begin{equation}
			\label{eq:applythm}
			\mathrm{tr}_H\circ{\mathcal{Y}}\circ{\mathcal{X}}=(\overline{\mathcal{X}}\otimes \overline{\mathcal{Y}})\circ\mathcal{D},
		\end{equation}
	where $\overline{\mathcal{X}}\circ\tr_J=\tr_H\circ\mathcal{X}$, $\overline{\mathcal{Y}}\circ\tr_I=\tr_H\circ\mathcal{Y}$.

	Next, we need to find the isometries that map the operators $X_{i,\alpha},Y_{j,\beta}$ on $K$ to the product Hilbert space $K\otimes K$ and the corresponding isometric operators. For this purpose, we first define CPTP maps $\overline{\mathcal{X}}_i,\overline{\mathcal{Y}}_j$ via
		\begin{align}
			\label{eq:defAi}
			\overline{\mathcal{X}}_i:&W_{K}\mapsto\sum_{\alpha}\ketbra{\alpha}{\alpha}_A\ \overline{\mathcal{X}}\big(\ketbra{i}{i}_{I}\otimes W_{K}\big)\ketbra{\alpha}{\alpha}_A\\
			\label{eq:defBj}
			\overline{\mathcal{Y}}_j:&W_{K}\mapsto\sum_{\beta}\ketbra{\beta}{\beta}_B\ \overline{\mathcal{Y}}\big(\ketbra{j}{j}_{J}\otimes W_{K}\big)\ketbra{\beta}{\beta}_B
		\end{align}
	and show that 
		\begin{equation}
			\label{eq:intermediateiso}
			\mathrm{tr}\big(X_{i,\alpha}Y_{j,\beta}\rho_K\big)=\big(\overline{\mathcal{X}}_{i,\alpha}\otimes\overline{\mathcal{Y}}_{j,\beta}\big)\circ\mathcal{D}(\rho_K),
		\end{equation}
	where $\overline{\mathcal{X}}_{i,\alpha}\coloneqq\bra{\alpha}\overline{\mathcal{X}}_i\ket{\alpha}$, $\overline{\mathcal{Y}}_{j,\beta}\coloneqq\bra{\beta}\overline{\mathcal{Y}}_j\ket{\beta}$. The proof goes as follows: 
	For any state $\rho_K$ on $K$ and all $i,j$, it follows from \cref{eq:defextAB} and commutativity that
	\begin{align}
		   \sum_{\alpha,\beta}\mathrm{tr}_K\big(X_{i,\alpha}Y_{j,\beta}\rho_K\big)\otimes\ketbra{\alpha}{\alpha}_A\otimes\ketbra{\beta}{\beta}_B
         	 & =
	  \mathrm{tr}_H \circ{\mathcal{Y}}\circ{\mathcal{X}}\big(\ketbra{i}{i}_{I}\otimes\rho_K\otimes\ketbra{j}{j}_{J}\big) \\
	  	& = 
	 (\overline{\mathcal{X}}\otimes \overline{\mathcal{Y}})\circ\mathcal{D} \big(\ketbra{i}{i}_{I}\otimes\rho_K\otimes\ketbra{j}{j}_{J}\big) \\
	& = \big(\overline{\mathcal{X}}\otimes\overline{\mathcal{Y}}\big)\, \big(\ketbra{i}{i}_{I}\otimes\mathcal{D}(\rho_K)\otimes\ketbra{j}{j}_{J}\big).
	\end{align}
	We may now apply the map $W_{A B} \mapsto \sum_{\tilde{\alpha}, \tilde{\beta}} \ketbra{\tilde{\alpha}}{\tilde{\alpha}} \otimes \ketbra{\tilde{\beta}}{\tilde{\beta}} W_{A B}  \ketbra{\tilde{\alpha}}{\tilde{\alpha}} \otimes \ketbra{\tilde{\beta}}{\tilde{\beta}}$ to the first and the last expression in this equality. Since this map acts like an identity on the first, we obtain 
		\begin{align}
			\sum_{\tilde{\alpha},\tilde{\beta}} & \mathrm{tr}_K\big(X_{i,\tilde{\alpha}}Y_{j,\tilde{\beta}}\rho_K\big)\otimes\ketbra{\tilde{\alpha}}{\tilde{\alpha}}_A\otimes\ketbra{\tilde{\beta}}{\tilde{\beta}}_B\\
			&=\sum_{\tilde{\alpha},\tilde{\beta}}\ketbra{\tilde{\alpha}}{\tilde{\alpha}}_A\otimes\ketbra{\tilde{\beta}}{\tilde{\beta}}_B\Big[\big(\overline{\mathcal{X}}\otimes\overline{\mathcal{Y}}\big)\, \big(\ketbra{i}{i}_{I}\otimes\mathcal{D}(\rho_K)\otimes\ketbra{j}{j}_{J}\big)\Big]\ketbra{\tilde{\alpha}}{\tilde{\alpha}}_A\otimes\ketbra{\tilde{\beta}}{\tilde{\beta}}_B\\
			&=(\overline{\mathcal{X}}_i\otimes \overline{\mathcal{Y}}_j)\circ\mathcal{D}(\rho_K),
		\end{align}
where we have used the definitions \cref{eq:defAi,eq:defBj}. Sandwiching this equality with $\ket{\alpha}_A\otimes \ket{\beta}_B$ yields \cref{eq:intermediateiso}, which we wanted to show.
	
	Next, note that $\overline{\mathcal{X}}_{i,\alpha}$ and $\overline{\mathcal{Y}}_{j,\beta}$ are CP maps from $K$ to a one-dimensional system. According to \cref{lem:1dmaps} there exist Hermitian operators $\overline{X}_{i,\alpha}$ and $\overline{Y}_{j,\beta}$ such that 
		\begin{align}
			\begin{split}
				\overline{\mathcal{X}}_{i,\alpha}(W_{K})&=\mathrm{tr}\left(\overline{X}_{i,\alpha} W_{K}\right)\\
				\overline{\mathcal{Y}}_{j,\beta}(W_{K})&=\mathrm{tr}\left(\overline{Y}_{j,\beta}W_{K}\right),\label{eq:HermitianOps}
			\end{split}
		\end{align}
	hence
		\begin{equation}
			\big(\overline{\mathcal{X}}_{i,\alpha}\otimes\overline{\mathcal{Y}}_{j,\beta}\big)(W_{KK})=\mathrm{tr}\big((\overline{X}_{i,\alpha}\otimes \overline{Y}_{j,\beta})W_{KK}\big).
		\end{equation}
	Thus, \cref{eq:intermediateiso} can be rewritten as
		\begin{align}
			\label{eq:rewrite}
			\mathrm{tr}\big(X_{i,\alpha}Y_{j,\beta}\rho_K\big)&=\mathrm{tr}\big((\overline{X}_{i,\alpha}\otimes\overline{Y}_{j,\beta})\circ\mathcal{D}(\rho_K)\big).
		\end{align}
	For the next step, we use that according to the Stinespring dilation, there exists an isometric map $\overline{\mathcal{D}}:K\to K\otimes K\otimes R$ such that $\tr_R\circ\overline{\mathcal{D}}=\mathcal{D}$, i.e.,  $\overline{\mathcal{D}}(\rho_K)={V}\rho_K{V}^*$ for some isometry ${V}:K\to K\otimes K\otimes R$. 
 Hence, \cref{eq:rewrite} can be written as
		\begin{align}
			\mathrm{tr}\big(X_{i,\alpha}Y_{j,\beta}\rho_K\big)&=\mathrm{tr}\big((\overline{X}_{i,\alpha}\otimes\overline{Y}_{j,\beta}\otimes\mathrm{id}_R)\circ\overline{\mathcal{D}}(\rho_K)\big)\\
			&=\mathrm{tr}\big((\overline{X}_{i,\alpha}\otimes\overline{Y}_{j,\beta}\otimes\mathrm{id}_R){V}\rho_K{V}^*\big)\\
			&=\mathrm{tr}\big({V}^*(\overline{X}_{i,\alpha}\otimes\overline{Y}_{j,\beta}\otimes\mathrm{id}_R){V}\rho_K\big).
		\end{align}
	This is true for any $\rho_K\in K$, hence
		\begin{equation}
			\label{eq:embedXY}
			X_{i,\alpha}Y_{j,\beta}={V}^*(\overline{X}_{i,\alpha}\otimes\overline{Y}_{j,\beta}\otimes\mathrm{id}_R){V}.
		\end{equation}
	
	Since $\overline{\mathcal{Y}}_j$ is trace-preserving, $\sum_{\beta}\overline{\mathcal{Y}}_{j,\beta}$ is also trace-preserving:
		\begin{equation}
			\mathrm{tr}\Big(\sum_{\beta}\overline{\mathcal{Y}}_{j,\beta}(W_K)\Big)=\sum_{\beta}\mathrm{tr}\big(\bra{\beta}\overline{\mathcal{Y}}_{j}(W_K)\ket{\beta}\big)=\mathrm{tr}\big(\overline{\mathcal{Y}}_{j}(W_K)\big)=\mathrm{tr}(W_K).
		\end{equation}
	Because this holds for all $W_K$ and all $j$, combining it with \cref{eq:HermitianOps} yields that the operators $\overline{Y}_{j,\beta}$ fulfil $\sum_{\beta}\overline{Y}_{j,\beta}=\mathrm{id}_{K}$ for all $j$. Similarly, we find that $\sum_{\alpha}\overline{X}_{i,\alpha}=\mathrm{id}_{K}$ for all $i$. Summing over $\beta$ in \cref{eq:embedXY} then yields
		\begin{equation}
			\label{eq:embedX}
			X_{i,\alpha}={V}^*\big(\overline{X}_{i,\alpha}\otimes\mathrm{id}_{K}\otimes\mathrm{id}_R\big){V}
		\end{equation}
 	and, similarly, summing over $\alpha$ yields
		\begin{equation}
			\label{eq:embedY}
			Y_{j,\beta}={V}^*\big(\mathrm{id}_{K}\otimes \overline{Y}_{j,\beta}\otimes\mathrm{id}_R\big){V}.
		\end{equation}
	
	With the identification $K_A\equiv K$, $K_B\equiv K\otimes R$, \cref{eq:embedX,eq:embedY} say that $X_{i,\alpha}$ and $Y_{j,\beta}$ are isometrically represented as operators on $\overline{K}=K_A\otimes K_B$ that act non-trivially only on $K_A$ and $K_B$, respectively.
\end{proof}

\begin{rem}
	It is actually not necessary in \cref{cor:vNalgebras} to assume that for all $i,j$, $\sum_\alpha X_{i,\alpha}=\mathrm{id}_K$ and $\sum_\beta Y_{j,\beta}=\mathrm{id}_K$. If this does not hold, we can scale the operators with a constant $\gamma>0$ such that $\sum_{\alpha} X_{i, \alpha} \le\frac{1}{\gamma}\mathrm{id}_K$ and add the operator $X_{i,0}\coloneqq \mathrm{id}_K-\gamma\sum_\alpha X_{i,\alpha}$ (and analogously for $Y$). This operator is also positive and commutes with all operators in $\{Y_{j,\beta}\}$.
\end{rem}

\begin{rem} \label{rem:TsirelsonGeneralised}
  As described above, \cref{cor:vNalgebras} is obtained from \cref{cor:commute} by treating $I$, $J$, $A$, and $B$ as classical systems. We could apply the same procedure directly to \cref{thm:tensorP}. This allows us to derive a stronger version of \cref{cor:vNalgebras}, where~\cref{eq:commutingoperators} is replaced by the weaker condition that the two families of operators $\{X_{i,\alpha}\}$ and $\{Y_{j,\beta}\}$ satisfy
    \begin{equation} \label{eq:weakcommutation}
    \sum_\alpha \sqrt{X_{i, \alpha}} Y_{j, \beta} \sqrt{X_{i, \alpha}}  =  Y_{j, \beta}   \quad  \text{ and } \forall \, i, j, \beta \ .
  \end{equation}
  Although the condition merely involves a sum over $\alpha$ rather than a commutation relation for each~$\alpha$, it suffices to imply that the operators factorise as in \cref{eq:operatorfactorisation}.
    
  To prove this, we define, analogously to \cref{eq:defextAB},
  \begin{align}
   			\begin{split}
			    \mathcal{M}&: W_{H}\mapsto\sum_{i,\alpha}\left(\ketbra{i}{i}_{I}\otimes\sqrt{X_{i,\alpha}}\otimes\mathrm{id}_{J}\right) W_{H}\left(\ketbra{i}{i}_{I}\otimes\sqrt{X_{i,\alpha}}\otimes\mathrm{id}_{J}\right)\otimes\ketbra{\alpha}{\alpha}_A,\\
			    \mathcal{N} &: W_{H}\mapsto\sum_{j,\beta} \mathrm{tr}_H \left(\mathrm{id}_{I}\otimes\sqrt{Y_{j,\beta}}\otimes\ketbra{j}{j}_{J}\right) W_{H}\left(\mathrm{id}_{I}\otimes\sqrt{Y_{j,\beta}}\otimes\ketbra{j}{j}_{J}\right)\otimes\ketbra{\beta}{\beta}_B.
			\end{split}
  \end{align}
  These CP maps manifestly satisfy Condition~\ref{cond:3} of \cref{thm:tensorP}. Furthermore, since $\mathcal{M}$ is identical to ${\mathcal{X}}$ as defined in the proof of \cref{cor:vNalgebras}, we already know that $\tr_A \circ \mathcal{M}$ is unital and trace-preserving, so Condition~\ref{cond:2} holds. Finally, Condition~\ref{cond:1} is equivalent to the requirement that, for all $\rho_K$ and all $i$, $j$,
    \begin{align}
       \sum_{\beta} \tr_A \circ \mathcal{N} \circ \mathcal{M}(\ketbra{i}{i} \otimes \rho_K \otimes \ket{j} \bra{j}) = \sum_{\beta} \mathcal{N}(\ketbra{i}{i} \otimes \rho_K \otimes \ket{j} \bra{j}).
    \end{align}
   Inserting the explicit expressions for the maps, the requirement can be rewritten as
   \begin{align} 
      \sum_{\alpha,\beta} \tr\big(Y_{j, \beta} \sqrt{X_{i, \alpha}} \rho_K \sqrt{X_{i, \alpha}}\big)\otimes\ketbra{\beta}{\beta}_B = \sum_{\beta}\tr\big(Y_{j, \beta} \rho_K\big)\otimes\ketbra{\beta}{\beta}_B \quad \forall \rho_K, i, j.
   \end{align}
  In particular, the equality must hold individually for each term of the sum over~$\beta$. It is thus equivalent to 
  \begin{align} 
  	\sum_{\alpha} \tr\big(Y_{j, \beta} \sqrt{X_{i, \alpha}} \rho_K \sqrt{X_{i, \alpha}}\big) = \tr\big(Y_{j, \beta} \rho_K\big) \quad \forall \rho_K, i, j,\beta,
  \end{align}
  which in turn is equivalent to \cref{eq:weakcommutation}.
  
\end{rem}

\begin{cor}
	\label{cor:multimap}
	Let $\mathcal{M}_1,\dots,\mathcal{M}_s$ be CPTP maps such that $\mathcal{M}_i:I_i \otimes K\to A_i \otimes K$, and $\tr_{A_i}\circ\mixed{I_i}\circ\mathcal{M}_i$ is unital.\footnote{Here, $\mixed{I_i}$ denotes the map that creates the corresponding state, see \cref{not:stategeneratingmap}.} If, for all $t\in\{1,\dots,s-1\}$,
		\begin{equation}
			\label{eq:multimapcond1}
			\tr_{A_1,\dots, A_t}\circ\tr_K\circ\mathcal{M}_s\circ\dots\circ\mathcal{M}_1=\tr_K\circ\mathcal{M}_s\circ\dots\circ\mathcal{M}_{t+1}\circ\tr_{I_t}\circ\dots\circ\tr_{I_1},
		\end{equation}
	then there exists a CPTP map $\mathcal{D}:K\to K\otimes K\otimes\dots\otimes K$ such that 
		\begin{equation}\label{eq:multimapresult}
			\tr_K\circ\mathcal{M}_s\circ \dots\circ\mathcal{M}_1=\left(\overline{\mathcal{M}}_1\otimes\dots\otimes\overline{\mathcal{M}}_s\right)\circ\mathcal{D},
		\end{equation}
	where $\overline{\mathcal{M}}_i=\tr_K\circ\mathcal{M}_i$.
\end{cor}

\begin{proof}
	
	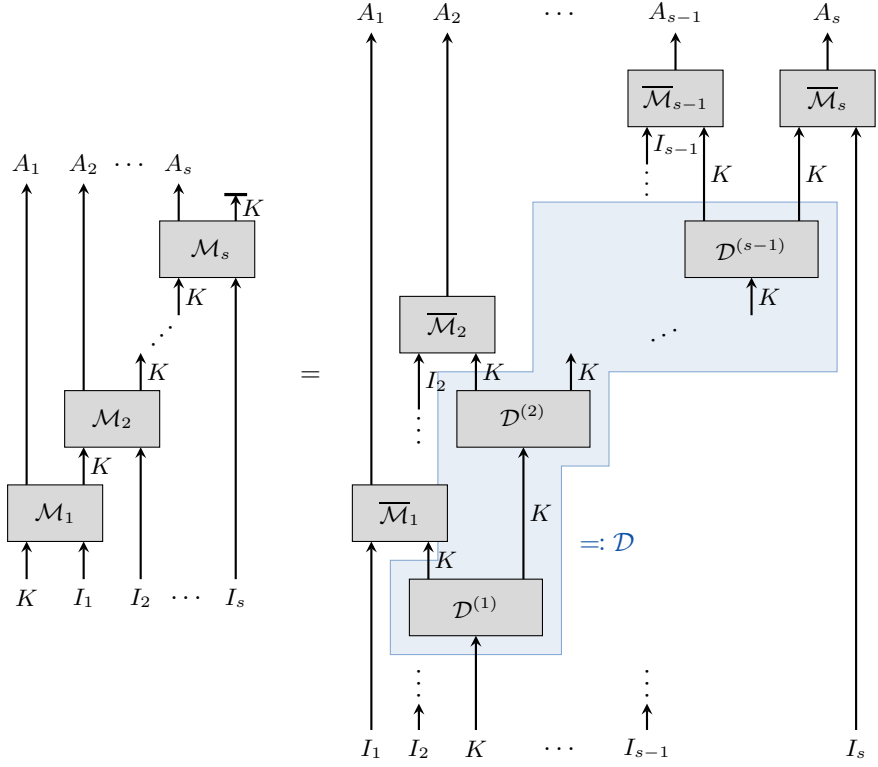
\begin{figure}[t]
		\centering
		\begin{tikzpicture}[baseline=(current bounding box.center)]
			\draw[fill=gray!30] (0,0) rectangle (1.25,0.75);
			\node at (0.625,0.375) {\small $\mathcal{M}_1$};
			\draw[fill=gray!30] (0.75,1.25) rectangle (2,2);
			\node at (1.375,1.625) {\small $\mathcal{M}_2$};
			\draw[fill=gray!30] (2,3.5) rectangle (3.25,4.25);
			\node at (2.675,3.875) {\small $\mathcal{M}_s$};
			\draw[->,>=stealth,thick] (0.25,-0.5) -- (0.25,0);
			\draw[->,>=stealth,thick] (1,-0.5) -- (1,0);
			\draw[->,>=stealth,thick] (0.25,0.75) -- (0.25,4.75);
			\draw[->,>=stealth,thick] (1,0.75) to node[right=-0.05cm] {\small $K$} (1,1.25);
			\draw[->,>=stealth,thick] (1,2) -- (1,4.75);
			\draw[->,>=stealth,thick] (1.75,2) to node[right=-0.05cm] {\small $K$} (1.75,2.5);
			\node[rotate=45] at (2.05,2.75) {$\dots$};
			\draw[->,>=stealth,thick] (2.25,3) to node[right=-0.05cm] {\small $K$} (2.25,3.5);
			\draw[->,>=stealth,thick] (3,-0.5) -- (3,3.5);
			\draw[->,>=stealth,thick] (2.25,4.25) -- (2.25,4.75);
			\draw[->,>=stealth,thick] (3,4.25) to node[right=-0.05cm] {\small $K$} (3,4.6);
			\draw[very thick] (2.85,4.6) -- (3.15,4.6);
			\draw[->,>=stealth,thick] (1.75,-0.5) -- (1.75,1.25);
			\node at (0.25,-0.75) {\small $K$};
			\node at (1,-0.75) {\small $I_1$};
			\node at (1.75,-0.75) {\small $I_2$};
			\node at (2.375,-0.75) {$\dots$};
			\node at (3,-0.75) {\small $I_s$};
			\node at (0.25,5) {\small $A_1$};
			\node at (1,5) {\small $A_2$};
			\node at (1.625,5) {$\dots$};
			\node at (2.25,5) {\small $A_s$};
		\end{tikzpicture}\ \ \ =\ \ \ 
		\begin{tikzpicture}[baseline=(current bounding box.center)]
			\draw[fill=myblue!10,draw=myblue!50] (0.5,-1.5) -- (2.75,-1.5) -- (2.75,1) -- (3.375,1) -- (3.375,2.25) -- (6.375,2.25) -- (6.375,4.5) -- (2.375,4.5) -- (2.375,2.25) -- (1.125,2.25) -- (1.125,-0.25) -- (0.5,-0.25) -- cycle;
			\node[color=myblue] at (3.35,0) {$\eqqcolon\mathcal{D}$};
			\draw[fill=gray!30] (0,0) rectangle (1.25,0.75);
			\node at (0.625,0.375) {\small $\overline{\mathcal{M}}_1$};
			\draw[fill=gray!30] (0.625,2.5) rectangle (1.875,3.25);
			\node at (1.25,2.875) {\small $\overline{\mathcal{M}}_2$};
			\draw[fill=gray!30] (0.75,-1.25) rectangle (2.5,-0.5);
			\node at (1.625,-0.875) {\small $\mathcal{D}^{(1)}$};
			\draw[fill=gray!30] (1.375,1.25) rectangle (3.125,2);
			\node at (2.25,1.675) {\small $\mathcal{D}^{(2)}$};
			\draw[->,>=stealth,thick] (0.25,-2.5) -- (0.25,0);
			\draw[->,>=stealth,thick] (1.625,-2.5) -- (1.625,-1.25);
			\draw[->,>=stealth,thick] (1,-0.5) to node[right=-0.05cm] {\small $K$} (1,0);
			\draw[->,>=stealth,thick] (2.25,-0.5) to node[right=-0.05cm] {\small $K$} (2.25,1.25);
			\draw[->,>=stealth,thick] (0.25,0.75) -- (0.25,6.75);
			\draw[->,>=stealth,thick] (1.625,2) to node[right=-0.05cm] {\small $K$} (1.625,2.5);
			\draw[->,>=stealth,thick] (2.875,2) to node[right=-0.05cm] {\small $K$} (2.875,2.5);
			\draw[->,>=stealth,thick] (0.875,1.75) to node[right=-0.05cm] {\small $I_2$} (0.875,2.5);
			\node[rotate=90] at (0.875,1.5) {$\dots$};
			\draw[->,>=stealth,thick] (0.875,-2.5) -- (0.875,-2.15);
			\node[rotate=90] at (0.875,-1.85) {$\dots$};
			\node[rotate=25] at (4.125,2.75) {$\dots$};
			\draw[->,>=stealth,thick] (1.25,3.25) -- (1.25,6.75);
			\draw[->,>=stealth,thick] (3.875,-2.5) -- (3.875,-2.15);
			\node[rotate=90] at (3.875,-1.85) {$\dots$};
			\begin{scope}[xshift=1.75cm,yshift=0.75cm]
			\draw[->,>=stealth,thick] (2.125,4.25) to node[right=-0.05cm,pos=0.5] {\small $I_{s-1}$} (2.125,4.75);
			\node[rotate=90] at (2.125,4.025) {\small $\dots$};
			\begin{scope}[yshift=-0.75cm]
			\draw[fill=gray!30] (2.625,3.5) rectangle (4.375,4.25);
			\node at (3.5,3.875) {\small $\mathcal{D}^{(s-1)}$};
			\draw[->,>=stealth,thick] (3.5,3) to node[right=-0.05cm] {\small $K$} (3.5,3.5);
			\end{scope}
			\draw[fill=gray!30] (1.875,4.75) rectangle (3.125,5.5);
			\node at (2.5,5.125) {\small $\overline{\mathcal{M}}_{s-1}$};
			\draw[fill=gray!30] (3.875,4.75) rectangle (5.125,5.5);
			\node at (4.5,5.125) {\small $\overline{\mathcal{M}}_{s}$};
			\node at (2.5,6.25) {\small $A_{s-1}$};
			\node at (4.5,6.25) {\small $A_s$};
			\node at (4.875,-3.5) {\small $I_s$};
			\draw[->,>=stealth,thick] (4.125,3.5) to node[right=-0.05cm] {\small $K$} (4.125,4.75);
			\draw[->,>=stealth,thick] (2.875,3.5) to node[right=-0.05cm] {\small $K$} (2.875,4.75);
			\draw[->,>=stealth,thick] (2.5,5.5) -- (2.5,6);
			\draw[->,>=stealth,thick] (4.5,5.5) -- (4.5,6);
			\draw[->,>=stealth,thick] (4.875,-3.25) -- (4.875,4.75);
			\end{scope}
			\node at (1.625,-2.75) {\small $K$};
			\node at (0.25,-2.75) {\small $I_1$};
			\node at (0.875,-2.75) {\small $I_2$};
			\node at (2.75,-2.75) {$\dots$};
			\node at (3.875,-2.75) {\small $I_{s-1}$};
			\node at (0.25,7) {\small $A_1$};
			\node at (1.25,7) {\small $A_2$};
			\node at (2.75,7) {$\dots$};
		\end{tikzpicture}
		\caption{\label{fig:multimapcorollary} \textbf{Visualisation of \cref{cor:multimap}.} The diagram shows~\cref{eq:multimapresult}, which generalises the statement of~\cref{thm:tensorP} to a sequence of $s$ maps. Here, $\overline{\mathcal{M}}_i=\tr_K\circ\mathcal{M}_i$ and $\mathcal{D}\coloneqq\mathcal{D}^{(s-1)}\circ\dots\circ\mathcal{D}^{(1)}$ (depicted in blue).}
	\end{figure}

	We will prove the statement via iteratively applying \cref{thm:tensorP}.
	First, note that we can always insert a map $\tr_{I_i}\circ\mixed{I_i}$ without changing the left-hand side of \cref{eq:multimapresult}, for example,
		\begin{equation}
			\tr_K\circ\mathcal{M}_s\dots\circ\mathcal{M}_1=\tr_K\circ\mathcal{M}_s\dots\circ\mathcal{M}_2\circ\tr_{I_1}\circ\mixed{I_1}\circ\mathcal{M}_1.
		\end{equation}
	Thus, for the first iteration, we choose $\mathcal{M}^{(1)}\coloneqq\mixed{I_1}\circ\mathcal{M}_1$ and $\mathcal{N}^{(1)}\coloneqq\tr_K\circ\mathcal{M}_s\circ\dots\circ\mathcal{M}_2\circ\tr_{I_1}$, as well as $I^{(1)}\coloneqq I_1$ and $J^{(1)}\coloneqq I_2\dots I_s$. With these definitions, all conditions in \cref{thm:tensorP} are fulfilled:
		\begin{enumerate}[label=(\roman*)]
			\item From \cref{eq:multimapcond1}, it directly follows that
			\begin{align}
				\tr_{A_1}\circ\mathcal{N}^{(1)}\circ\mathcal{M}^{(1)}&=\tr_{A_1}\circ\left(\tr_K\circ\mathcal{M}_s\circ\dots\circ\mathcal{M}_2\circ\tr_{I_1}\right)\circ\left(\mixed{I_1}\circ\mathcal{M}_1\right)\\
				&=\tr_{A_1}\circ\tr_K\circ\mathcal{M}_s\circ\dots\circ\mathcal{M}_1\\
				&=\tr_K\circ\mathcal{M}_s\circ\dots\circ\mathcal{M}_2\circ\tr_{I_1}\\
				&=\mathcal{N}^{(1)}.
			\end{align}
			\item $\tr_{A_1}\circ\mathcal{M}^{(1)}=\tr_{A_1}\circ\mixed{I_1}\circ\mathcal{M}_1$ is unital by assumption.
			\item $\tr_H\circ\mathcal{M}^{(1)}=\tr_H\circ\mixed{I_1}\circ\mathcal{M}_1=\tr_K\circ\mathcal{M}_1\circ\tr_{I_2\dots I_s}$ and $\mathcal{N}^{(1)}=\tr_K\circ\mathcal{M}_s\circ\dots\circ\mathcal{M}_2\circ\tr_{I_1}$ are obviously independent of $J^{(1)}=I_2\dots I_s$ and $I^{(1)}=I_1$, respectively.
		\end{enumerate}
	Thus, we can apply \cref{thm:tensorP}, which implies the existence of a CPTP map $\mathcal{D}^{(1)}:K\to K\otimes K$ such that
		\begin{equation}
			\label{eq:multimap1stiteration}
			\mathcal{N}^{(1)}\circ\mathcal{M}^{(1)}=\left(\overline{\mathcal{M}^{(1)}}\otimes\overline{\mathcal{N}^{(1)}}\right)\circ\mathcal{D}^{(1)},
		\end{equation}
	where $\overline{\mathcal{M}^{(1)}}\coloneqq\tr_K\circ\mathcal{M}_1$ and $\overline{\mathcal{N}^{(1)}}\coloneqq\tr_K\circ\mathcal{M}_s\circ\dots\mathcal{M}_2$. Thus, \cref{eq:multimap1stiteration} translates to
		\begin{equation} \label{eq:firstinductionresult}
			\tr_K\circ\mathcal{M}_s\circ\dots\circ\mathcal{M}_1=\Bigl(\bigl(\underbrace{\tr_K\circ\mathcal{M}_1}_{=\overline{\mathcal{M}^{(1)}}}\bigr)\otimes\bigl(\underbrace{\tr_K\circ\mathcal{M}_s\circ\dots\circ\mathcal{M}_2}_{=\overline{\mathcal{N}^{(1)}}}\bigr)\Bigr) \circ\mathcal{D}^{(1)}.
		\end{equation}
	
	Iteration steps $2$ to $s-1$ then work similarly: First, we show that \cref{eq:multimapcond1} implies that for all $t\in\{1,\dots,s-1\}$,
		\begin{equation}\label{eq:multimapcondt}
			\tr_{A_{t}}\circ\tr_K\circ\mathcal{M}_s\circ\dots\circ\mathcal{M}_{t}=\tr_K\circ\mathcal{M}_s\circ\dots\circ\mathcal{M}_{t+1}\circ\tr_{I_t}.
		\end{equation} 
	This can be derived by applying \cref{eq:multimapcond1} twice:
		\begin{align}
			\tr_{A_t}\circ\tr_K\circ\mathcal{M}_s\circ\dots\circ\mathcal{M}_{t}\circ\tr_{I_{t-1}\dots I_1}&=\tr_{A_t} \circ\tr_{A_{t-1}\dots A_1}\circ\tr_K\circ\mathcal{M}_s\circ\dots\circ\mathcal{M}_1\\
			&=\tr_K\circ\mathcal{M}_s\circ\dots\circ\mathcal{M}_{t+1}\circ\tr_{I_t}\circ\tr_{I_{t-1}\dots I_1}.
		\end{align}
	We now set $\mathcal{M}^{(t)}\coloneqq\mixed{I_t}\circ\mathcal{M}_t$ and $\mathcal{N}^{(t)}\coloneqq\tr_K\circ\mathcal{M}_s\circ\dots\circ\mathcal{M}_{t+1}\circ\tr_{I_t}$, and $I^{(t)}\coloneqq I_t$, $J^{(t)}\coloneqq I_{t+1}\dots I_s$. For this choice, the assumptions of \cref{thm:tensorP} are fulfilled: $\tr_{A_t}\circ\mixed{I_t}\circ\mathcal{M}_t$ is unital by assumption, and it is clear that $\mathcal{M}^{(t)}$ and $\mathcal{N}^{(t)}$ are independent of $I_{t+1}\dots I_s$ and $I_t$, respectively. Furthermore, \cref{eq:multimapcondt} implies that Condition~\ref{cond:1} is fulfilled:
		\begin{align}
			\tr_{A_t}\circ\mathcal{N}^{(t)}\circ\mathcal{M}^{(t)}&=\tr_{A_t}\circ\left(\tr_K\circ\mathcal{M}_s\circ\dots\circ\mathcal{M}_{t+1}\circ\tr_{I_t}\right)\circ\left(\mixed{I_t}\circ\mathcal{M}_t\right)\\
			&=\tr_{A_t}\circ\tr_K\circ\mathcal{M}_s\circ\dots\circ\mathcal{M}_t\\
			&=\tr_K\circ\mathcal{M}_s\circ\dots\circ\mathcal{M}_{t+1}\circ\tr_{I_t}\\
			&=\mathcal{N}^{(t)}.
		\end{align}
	Hence, \cref{thm:tensorP} yields that there exists a CPTP map $\mathcal{D}^{(t)}:K\to K\otimes K$ such that 
		\begin{equation}
		        \overline{\mathcal{N}^{(t-1)}}
		        = \mathcal{N}^{(t)}\circ\mathcal{M}^{(t)}=\left(\overline{\mathcal{M}^{(t)}}\otimes\overline{\mathcal{N}^{(t)}}\right)\otimes\mathcal{D}^{(t)},
		\end{equation}
	where $\overline{\mathcal{M}^{(t)}}\coloneqq\tr_K\circ\mathcal{M}_t$ and $\overline{\mathcal{N}^{(t)}}\coloneqq\tr_K\circ\mathcal{M}_s\circ\dots\circ\mathcal{M}_{t+1}$. Starting from \cref{eq:firstinductionresult} and using this induction step repeatedly, we obtain
		\begin{equation}
			\tr_K\circ\mathcal{M}_s \circ \dots\circ\mathcal{M}_1=\left(\overline{\mathcal{M}}_1\otimes\dots\otimes\overline{\mathcal{M}}_s\right)\circ\mathcal{D},
		\end{equation}
	where $\overline{\mathcal{M}}_i=\tr_K\circ\mathcal{M}_i$ and $\mathcal{D}=\mathcal{D}^{(s-1)}\circ\dots\circ\mathcal{D}^{(1)}$ (see \cref{fig:multimapcorollary}).
	
\end{proof}

\begin{rem}
  In the same way as \cref{cor:commute} replaces Condition~\ref{cond:1} of \cref{thm:tensorP}  by a commutation condition, one may replace assumption \cref{eq:multimapcond1} of \cref{cor:multimap} on the maps $\mathcal{M}_1, \ldots, \mathcal{M}_s$ by a commutation assumption, namely that changing the order of the maps in the concatenation $\mathrm{tr}_K \circ \mathcal{M}_s \circ \cdots \circ \mathcal{M}_1$ does not have an effect. This results in a statement similar to \cref{cor:commute}: If the $s \geq 2$ maps on the left-hand side of \cref{eq:multimapresult} commute and satisfy the unitality condition, then they factorise as on the right-hand side of \cref{eq:multimapresult}.
\end{rem}

%% file: A_Choi.tex
\section{On the Choi-Jamio\l kowski isomorphism}
\label{app:CJ}

Let $\psi_{H\tilde{H}}=\smash{\ketbra{\psi}{\psi}_{H\tilde{H}}}$ be a maximally entangled state between a finite-dimensional space $H$ and an isomorphic space $\tilde{H}$, which we keep fixed for the following discussion. According to \cref{rem:CJ}, we may choose an orthonormal basis $\smash{\{\ket{i}_H\}_{i \in \{1, \ldots, d\}}}$ of $H$, which then induces an orthonormal basis $\{\ket{i}_{\tilde{H}}\}_{i \in \{1, \ldots d\}}$ on $\tilde{H}$ such that
	\begin{equation}
		\label{eq:Schmidt}
		\ket{\psi}_{H\tilde{H}}=\frac{1}{\sqrt{d}}\sum_i\ket{i}_H \otimes \ket{i}_{\tilde{H}},
	\end{equation}
where $d\coloneqq\dim(H) = \dim(\tilde{H})$. 

\begin{defn}
	\label{def:CJstate}
	 The Choi-Jamio\l kowski (C.-J.) operator of a CP map $\mathcal{M}:H\to K$ is defined as 
		\begin{equation}
			\rho_{K\tilde{H}}\coloneqq\mathcal{M}(\psi_{H\tilde{H}}).
		\end{equation}
\end{defn}

\begin{lem} \label{lem:CJrepresentation}
	Let $\rho_{K\tilde{H}}$ be the C.-J.~operator of a CP map $\mathcal{M}:H\to K$. Then for all operators $W_H$ on $H$,
		\begin{equation}
			\mathcal{M}(W_H)=d^2\,\mathrm{tr}_{\tilde{H}}\Big(\mathrm{tr}_H\big(W_H\psi_{H\tilde{H}}\big)\rho_{K\tilde{H}}\Big).
		\end{equation} 
\end{lem}

\begin{proof}
	Using \cref{def:CJstate} and \cref{eq:Schmidt}, we can verify the claim by a direct calculation (states with the same colour yield a $\delta$-function of the corresponding labels):
		\begin{align}
			d^2\,\mathrm{tr}_{\tilde{H}}&\Big(\mathrm{tr}_H\big(W_H\psi_{H\tilde{H}}\big)\rho_{K\tilde{H}})\Big)\\
			&=d^2\,\mathrm{tr}_{\tilde{H}}\Big(\mathrm{tr}_H\big(W_H\psi_{H\tilde{H}}\big)\mathcal{M}(\psi_{H\tilde{H}})\Big)\\
			&=\sum_{i,i',j,j'}\mathrm{tr}_{\tilde{H}}\Big(\mathrm{tr}_H\big(W_H(\ket{i}_H\ket{i}_{\tilde{H}}\bra{i'}_H\bra{i'}_{\tilde{H}})\big)\mathcal{M}(\ket{j}_H\ket{j}_{\tilde{H}}\bra{j'}_H\bra{j'}_{\tilde{H}})\Big)\\
			&=\sum_{k,\tilde{k}}\sum_{i,i',j,j'}{\color{mygreen}\bra{\tilde{k}}_{\tilde{H}}}\Big(\bra{k}_H\big(W_H(\ket{i}_H{\color{mygreen}\ket{i}_{\tilde{H}}}{\color{myorange}\bra{i'}_H}{\color{myblue}\bra{i'}_{\tilde{H}}}){\color{orange}\ket{k}_H}\big)\mathcal{M}(\ket{j}_H{\color{myblue}\ket{j}_{\tilde{H}}}\bra{j'}_H{\color{myred}\bra{j'}_{\tilde{H}}})\Big){\color{myred}\ket{\tilde{k}}_{\tilde{H}}}\\
			&=\sum_{i,j}\bra{j}_HW_H\ket{i}_H\mathcal{M}(\ket{j}_H\bra{i}_H)\\
			&=\mathcal{M}\Big(\sum_{i,j} \ket{j}_H \bra{j}_H W_H\ket{i}_H\bra{i}_H\Big)\\
			&=\mathcal{M}\big(W_H\big).
		\end{align}
\end{proof}

\begin{lem} \label{lem:CJtracepreserving}
	A CP map $\mathcal{M}:H\to K$ is trace non-increasing if and only if its corresponding C.-J.~operator $\rho_{K\tilde{H}}$ fulfils $\rho_{\tilde{H}}\le\mixed{\tilde{H}}$. It is trace-preserving if and only if $\rho_{\tilde{H}}=\mixed{\tilde{H}}$.
\end{lem}

\begin{proof}
	Let 
	\begin{align}
	  \mathcal{M}(W_H) = \sum_z E_z W_H E_z^* 
	\end{align} 
	be the Kraus representation of $\mathcal{M}$ with Kraus operators $E_z$. We then have
	\begin{align}
	    \rho_{\tilde{H}}= \mathrm{tr}_K\bigl(\mathcal{M}(\psi_{H\tilde{H}})\bigr)
	    & = \sum_z \mathrm{tr}_K\bigl((E_z \otimes \mathrm{id}_{\tilde{H}}) \psi_{H\tilde{H}} (E_z^* \otimes \mathrm{id}_{\tilde{H}})\bigr) \\
	    & =  \sum_z \mathrm{tr}_H\bigl( \psi_{H\tilde{H}} (E_z^* E_z \otimes \mathrm{id}_{\tilde{H}}) \bigr) \\
	    & = \mathrm{tr}_H\bigl(\psi_{H\tilde{H}}  \sum_z E_z^* E_z \otimes \mathrm{id}_{\tilde{H}}  \bigr),
	\end{align}
	where we used the cyclicity of the trace. Because, for trace non-increasing maps, $\sum_x E_x^* E_x \leq \id_{H}$, we find
	\begin{align}
	  \rho_{\tilde{H}} \leq  \mathrm{tr}_H\bigl(\psi_{H\tilde{H}} ( \mathrm{id}_H \otimes \mathrm{id}_{\tilde{H}})  \bigr) 
	  = \mathrm{tr}_H\bigl(\psi_{H\tilde{H}}\bigr) = \mixed{\tilde{H}}.
	\end{align}		
	The inequality becomes an equality if $\mathcal{M}$ is trace-preserving.
	
	To show the other direction, suppose $\rho_{\tilde{H}}\le\mixed{\tilde{H}}$. Then, using \cref{lem:CJrepresentation},
		\begin{align}
			\mathrm{tr}_{K}\big(\mathcal{M}(W_H)\big)
			&=d^2\, \mathrm{tr}_K\Big[\mathrm{tr}_{\tilde{H}}\Big(\mathrm{tr}_H(W_H\psi_{H\tilde{H}})\rho_{K\tilde{H}}\Big)\Big]\\
			&=d^2\, \mathrm{tr}_{\tilde{H}}\Big(\mathrm{tr}_H(W_H\psi_{H\tilde{H}})\underbrace{\mathrm{tr}_K(\rho_{K\tilde{H}})}_{\le \mixed{\tilde{H}}}\Big)\\
			&\le d\, \mathrm{tr}_{\tilde{H}}\big(\mathrm{tr}_H(W_H\psi_{H\tilde{H}})\big) \label{eq:le_id2}\\
			&= d\, \mathrm{tr}_H\big(W_H\underbrace{\mathrm{tr}_{\tilde{H}}(\psi_{H\tilde{H}}}_{=\mixed{H}})\big)\\
			&=\mathrm{tr}_H(W_H).
		\end{align}
	The inequality in \cref{eq:le_id2} becomes an equality if $\rho_{\tilde{H}}=\mixed{\tilde{H}}$.
\end{proof}

\begin{lem}
	\label{lem:productform}
	A CP map $\mathcal{M}:H_A \otimes H_B \to K_A \otimes K_B$, with $H_A \otimes H_B$ finite-dimensional, has product form $\mathcal{M}^{(A)}\otimes\mathcal{M}^{(B)}$, where $\mathcal{M}^{(A)}:H_A\to K_A$ and $\mathcal{M}^{(B)}:H_B\to K_B$ if and only if its C.-J.~operator $\rho_{K_A K_B \tilde{H}_A \tilde{H}_B}$ has product form $\rho_{K_A \tilde{H}_A} \otimes \rho_{K_B \tilde{H}_B}$.
\end{lem}

\begin{proof}
 According to \cref{rem:CJ}, the entangled state $\ket{\psi}_{H \tilde{H}}$ used for the C.-J.~isomorphism can be decomposed as
\begin{align} \label{eq:maximalentangledproduct}
  \ket{\psi}_{H_AH_B\tilde{H}_A\tilde{H}_B}=\ket{\psi}_{H_A\tilde{H}_A}\otimes\ket{\psi}_{H_B\tilde{H}_B},
\end{align}
where $\ket{\psi}_{H_A\tilde{H}_A}$ and $\ket{\psi}_{H_B\tilde{H}_B}$ are maximally entangled states on the respective subsystems.  Suppose now that $\mathcal{M}=\mathcal{M}^{(A)}\otimes\mathcal{M}^{(B)}$. The corresponding C.-J.~operator is then given by 
	\begin{align}
		\rho_{K_AK_B\tilde{H}_A\tilde{H}_B}&=\mathcal{M}\big(\psi_{H_AH_B\tilde{H}_A\tilde{H}_B}\big)\\
		&=\mathcal{M}^{(A)}\otimes\mathcal{M}^{(B)}\big(\psi_{H_A\tilde{H}_A}\otimes\psi_{H_B\tilde{H}_B}\big)\\
		&=\mathcal{M}^{(A)}\big(\psi_{H_A\tilde{H}_A}\big)\otimes\mathcal{M}^{(B)}\big(\psi_{H_B\tilde{H}_B}\big)\\
		&\eqqcolon\rho_{K_A\tilde{H}_A}\otimes\rho_{K_B\tilde{H}_B},
	\end{align} 
	hence it has product form. 
		
	For the other direction, suppose that the C.-J.~operator is of the form $\rho_{K_AK_B\tilde{H}_A\tilde{H}_B}=\rho_{K_A\tilde{H}_A}\otimes\rho_{K_B\tilde{H}_B}$. The corresponding map is then given by 
	\begin{align}
		&\mathcal{M}(W_{H_AH_B})\\
		&=d^2\,\mathrm{tr}_{\tilde{H}_A\tilde{H}_B}\Big(\mathrm{tr}_{H_AH_B}\big(W_{H_AH_B}\psi_{H_AH_B\tilde{H}_A\tilde{H}_B}\big)\rho_{K_AK_B\tilde{H}_A\tilde{H}_B}\Big)\\
		&=(\dim(H_A)\dim(H_B))^2\,\mathrm{tr}_{\tilde{H}_A\tilde{H}_B}\Big(\mathrm{tr}_{H_AH_B}\big(W_{H_AH_B}\psi_{H_A\tilde{H}_A}\otimes\psi_{H_B\tilde{H}_B}\big)\rho_{K_A\tilde{H}_A}\otimes\rho_{K_B\tilde{H}_B}\Big)\\
		&=\Big(\dim(H_A)^2\tr_{\tilde{H}_AH_A}\big(\psi_{H_A\tilde{H}_A}\rho_{K_A\tilde{H}_A}\big)\otimes\dim(H_B)^2\tr_{\tilde{H}_BH_B}\big(\psi_{H_B\tilde{H}_B}\rho_{K_B\tilde{H}_B}\big)\Big)(W_{H_AH_B})\\
		&\eqqcolon\big(\mathcal{M}^{(A)}\otimes\mathcal{M}^{(B)}\big)(W_{H_AH_B}),
	\end{align}
	hence it has product form.
\end{proof}

%% file: B_unital.tex
\section{Necessity of the unitality condition}
\label{app:unitality}

In \cref{thm:tensorP}, the unitality condition (Condition~\ref{cond:2}; see also the slightly weaker requirement mentioned in \cref{rem:unitalrequirement}), is necessary (see \cref{thm:converse}). The same condition also occurs in \cref{cor:commute}. Its necessity  can be illustrated with the following example of maps for \cref{cor:commute}.

Let $A$, $B$, and $H = (I, J, K)$ be classical random variables, where $I$, $J$, and $K$ take values $i,j\in\{0,1\}$ and $k\in\{0,1\}^{2}\cup\{\perp\}$. Define the maps ${\mathcal{X}}:H\to A\otimes H$ and ${\mathcal{Y}}:H\to B\otimes H$ as follows.
	\begin{align}
		{\mathcal{X}}:\hspace{10pt}&\text{if }k=\perp\text{ then }a\in_R\{0,1\}; (k_1,k_2)\coloneqq(a,i)\\
		&\text{else }a\coloneqq k_1\oplus(i\cdot k_2) \\
		{\mathcal{Y}}:\hspace{10pt}&\text{if }k=\perp\text{ then }b\in_R\{0,1\}; (k_1,k_2)\coloneqq(b,j)\\
		&\text{else }b\coloneqq k_1\oplus(j\cdot k_2) 
	\end{align}
where $a\in_R\{0,1\}$ means that $a$ is chosen uniformly at random.

These maps fulfil all conditions in \cref{cor:commute} except for the unitality condition: Firstly, they are CP and TP since they are defined in terms of functions of variables. Furthermore, it is clear from their definitions that $\mathrm{tr}_{H} \circ {\mathcal{X}}$ and $\mathrm{tr}_H \circ {\mathcal{Y}}$ are independent of $J$ and $I$, respectively, i.e., Condition~\ref{cond:c3} holds. Finally, they commute: This is obvious for the input $(i,j,k\neq\perp)$ since in this case, the maps do not change the value of $k$. For $(i,j,k=\perp)$, the output on $A$ and $B$ of ${\mathcal{Y}}\circ{\mathcal{X}}$ is $a=r$ and $b=r\oplus (j\cdot i)$, where $r$ is a uniform random bit. Hence, $a$ and $b$ are both uniformly random bits with the correlation $a\oplus b=i\cdot j$. On the other hand, the output of ${\mathcal{X}}\circ{\mathcal{Y}}$ is $b=r$ and $a=r\oplus (i\cdot j)$. Again, this means that $a$ and $b$ are both uniform random bits with the correlation $a\oplus b=i\cdot j$. Thus, the probability distribution $\mathrm{Pr}_{AB|IJK}(\cdot,\cdot|i,j,\perp)$ is identical for ${\mathcal{Y}}\circ{\mathcal{X}}$ and ${\mathcal{X}}\circ{\mathcal{Y}}$, hence the maps satisfy Condition~\ref{cond:c1} of \cref{cor:commute}. 

However, $\mathrm{tr}_{A}\circ{\mathcal{X}}$ is not unital, i.e., uniform distributions are not mapped to uniform distributions. This can be seen from the fact that the output on $K$ is never $k=\perp$, hence the probability of $k=\perp$ is zero. In particular, the probability distribution of the output $K$ is never uniform. By symmetry, $\mathrm{tr}_{B}\circ{\mathcal{Y}}$ is not unital, either. Hence, Condition~\ref{cond:c2} is violated.

The question is now whether it is still possible to find a CPTP map $\mathcal{D}$ such that $\mathrm{tr}_H\circ{\mathcal{Y}}\circ{\mathcal{X}}$ can be written as $(\overline{\mathcal{X}}\otimes\overline{\mathcal{Y}})\circ\mathcal{D}$ as in \cref{cor:commute}, even though the unitality condition is not fulfilled. This is equivalent to asking whether the input-output behaviour described above can be generated with a setup as depicted on the right-hand side of \cref{fig:corollary}. Note that this setup corresponds to that of a CHSH game (i.e., a Bell test), where Alice and Bob each have local inputs $I$ and $J$, respectively, and outputs $A$ and $B$. Furthermore, each of them has access to one part of a bipartite quantum system $K \otimes K$, which may, for example, be prepared in a maximally entangled state. Here, Alice's output $A$ is independent of Bob's input $J$, and Bob's output $B$ is independent of Alice's input $I$. 

Crucially, the input-output behaviour defined by $\mathrm{Pr}_{AB|IJK}$ for $k=\perp$ corresponds to a PR box \cite{PopescuRohrlich94}, whose characteristics is that it always fulfils the winning condition, $a\oplus b=i\cdot j$, of the CHSH game. However, from \cite{Tsirelson1980} we know that this condition can only be fulfilled with a probability $\approx 85\%$ while \cref{cor:commute} would imply that it is fulfilled with certainty if it were applicable. 

We conclude that the unitality condition is necessary for \cref{cor:commute}, i.e., it cannot be dropped without replacement. Since the corollary is an implication of \cref{thm:tensorP}, this also shows that the unitality condition is necessary for \cref{thm:tensorP}, i.e., the claim of the theorem would be wrong if Condition~\ref{cond:2} were dropped---a fact that also follows from our converse statement, \cref{thm:converse}.

%% file: Theorem.bbl
\newcommand{\etalchar}[1]{$^{#1}$}
\begin{thebibliography}{HJPW04}
\expandafter\ifx\csname url\endcsname\relax
  \def\url#1{\texttt{#1}}\fi
\expandafter\ifx\csname doi\endcsname\relax
  \def\doi#1{\burlalt{#1}{https://dx.doi.org/#1}}\fi
\expandafter\ifx\csname urlprefix\endcsname\relax\def\urlprefix{}\fi
\expandafter\ifx\csname href\endcsname\relax
  \def\href#1#2{#2}\fi
\expandafter\ifx\csname burlalt\endcsname\relax
  \def\burlalt#1#2{\href{#2}{#1}}\fi

\bibitem[Cho75]{Choi1975}
Man-Duen Choi.
\newblock Completely positive linear maps on complex matrices.
\newblock {\em Linear Algebra and its Applications}, 10(3):285--290, 1975,
  \doi{10.1016/0024-3795(75)90075-0}.

\bibitem[CQK23]{Cabello2023}
Adán Cabello, Marco~Túlio Quintino, and Matthias Kleinmann.
\newblock Logical possibilities for physics after {MIP*=RE}.
\newblock 2023.
\newblock \burlalt{arXiv:2307.02920}{https://arxiv.org/abs/2307.02920}.

\bibitem[HJPW04]{Hayden2004}
Patrick Hayden, Richard Jozsa, Dénes Petz, and Andreas Winter.
\newblock Structure of states which satisfy strong subadditivity of quantum
  entropy with equality.
\newblock {\em Communications in Mathematical Physics}, 246(2):359--374, 2004,
  \doi{10.1007/s00220-004-1049-z}.
\newblock
  \burlalt{arXiv:quant-ph/0304007}{https://arxiv.org/abs/quant-ph/0304007}.

\bibitem[Jam72]{Jamiolkowski1972}
Andrzej Jamio{\l}kowski.
\newblock Linear transformations which preserve trace and positive
  semidefiniteness of operators.
\newblock {\em Reports on Mathematical Physics}, 3(4):275--278, 1972,
  \doi{10.1016/0034-4877(72)90011-0}.

\bibitem[JNV{\etalchar{+}}21]{Ji2021}
Zhengfeng Ji, Anand Natarajan, Thomas Vidick, John Wright, and Henry Yuen.
\newblock {MIP}$^\ast$ = {RE}.
\newblock {\em Communications of the {ACM}}, 64(11):131--138, 2021,
  \doi{10.1145/3485628}.
\newblock \burlalt{arXiv:2001.04383}{https://arxiv.org/abs/2001.04383}.

\bibitem[LB21]{Lorenz2021}
Robin Lorenz and Jonathan Barrett.
\newblock Causal and compositional structure of unitary transformations.
\newblock {\em Quantum}, 5:511, 2021, \doi{10.22331/q-2021-07-28-511}.
\newblock \burlalt{arXiv:2001.07774}{https://arxiv.org/abs/2001.07774}.

\bibitem[PR94]{PopescuRohrlich94}
Sandu Popescu and Daniel Rohrlich.
\newblock Quantum nonlocality as an axiom.
\newblock {\em Foundations of Physics}, 24(3):379--385, 1994,
  \doi{10.1007/BF02058098}.

\bibitem[SW08]{Scholz2008}
Volkher~B. Scholz and Reinhard~F. Werner.
\newblock Tsirelson's problem.
\newblock 2008.
\newblock \burlalt{arXiv:0812.4305}{https://arxiv.org/abs/0812.4305}.

\bibitem[Tsi80]{Tsirelson1980}
Boris~S. Tsirelson.
\newblock Quantum generalizations of {Bell}'s inequality.
\newblock {\em Letters in Mathematical Physics}, 4(2):93--100, 1980,
  \doi{10.1007/bf00417500}.

\bibitem[Tsi93]{Tsirelson1993}
Boris~S. Tsirelson.
\newblock Some results and problems on quantum {Bell}-type inequalities.
\newblock {\em Hadronic Journal Supplement}, 8:329--345, 1993.
\newblock \url{http://www.tau.ac.il/~tsirel/download/hadron.html}.

\bibitem[Tsi06]{Tsirelson2006}
Boris~S. Tsirelson.
\newblock Bell inequalities and operator algebras, 2006.
\newblock \url{https://www.tau.ac.il/~tsirel/download/bellopalg.pdf}.

\end{thebibliography}
